\documentclass[journal,10pt]{IEEEtran}
\usepackage{setspace}
\usepackage{amsfonts}
\usepackage{epstopdf}
\usepackage{indentfirst}
\usepackage{amsmath,amsthm}
\usepackage{amssymb}
\usepackage{cite} 
\usepackage{color}
\usepackage{graphicx}
\usepackage{adjustbox}
\usepackage{multirow,makecell,diagbox}
\usepackage[para]{threeparttable}
\usepackage{tabularx}
\usepackage{rotating}

\definecolor{burntorange}{rgb}{0.8, 0.33, 0.0}
\newtheorem{Theorem}{Theorem}
\newtheorem{Corollary}{Corollary}
\newtheorem{Definition}{Definition}

\newtheorem{Remark}{Remark}
\newtheorem{Lemma}{Lemma}
\newtheorem{Example}{Example}

\begin{document}
\title{Symmetrical Z-Complementary Code Sets for Optimal Training in Generalized Spatial Modulation}

\author{Yajing Zhou,~\IEEEmembership{Student Member,~IEEE,}
\thanks{Y. Zhou is with the School of Information Science and Technology, Southwest Jiaotong University,
Chengdu, 610031, China  (e-mail:
zhouyajing@my.swjtu.edu.cn).}
Zhengchun Zhou,~\IEEEmembership{Member,~IEEE,}\thanks{ Z. Zhou and Y. Yang are with the School of Mathematics, Southwest Jiaotong University,
Chengdu, 610031, China (e-mail:
zzc@swjtu.edu.cn, yang\_data@swjtu.edu.cn).}
Zilong Liu,~\IEEEmembership{Senior Member,~IEEE,}
\thanks{Z. Liu is with the School of Computer Science and Electronics Engineering, University of Essex,
Colchester, CO4 3SQ, United Kingdom  (e-mail:
zilong.liu@essex.ac.uk).}
Yang Yang,~\IEEEmembership{Member,~IEEE,}
Ping Yang, ~\IEEEmembership{Senior Member,~IEEE,}\thanks{P. Yang is with the National Key Laboratory on Communications, University of Electronic Science and Technology of China, Chengdu, 611731, China
(e-mail: yang.ping@uestc.edu.cn).}
Pingzhi Fan,~\IEEEmembership{Fellow,~IEEE,}\thanks{P. Fan is with the Institute of Mobile Communications, Southwest Jiaotong
University, Chengdu, 611756, China. E-mail: pzfan@swjtu.edu.cn.}}
\maketitle

\date{}
\begin{abstract}
This paper introduces a novel class of code sets, called ``symmetrical Z-complementary code sets'', whose aperiodic auto- and cross- correlation sums exhibit zero-correlation zones at both the front-end and tail-end of the entire correlation window. Two constructions of (optimal) symmetrical Z-complementary code sets based on generalized Boolean functions are presented. Furthermore, we apply symmetrical Z-complementary code sets to design optimal training sequences for broadband generalized spatial modulation systems over frequency-selective channels.

\end{abstract}

\section{Introduction}

This paper is an extension of an emerging concept called ``cross Z-complementa- ry pair (CZCP)'' \cite{Liu-Yang-20} which has been proposed for optimal training in spatial modulation (SM) systems. Unlike \cite{Liu-Yang-20} focusing on training design for SM  in which only one transmit antenna is activated at each symbol slot, we aim to attain optimal training in generalized spatial modulation (GSM) systems where two or more transmit antennas may be activated. Towards this objective, we propose and construct a novel class of codes called symmetrical Z-complementary code set (SZCCS). In the sequel, we overview a number of code sets, introduce the concept of GSM, followed by our major contributions of this work.

\subsection{Complementary Code Sets}

Sequence sets with good correlation properties play an instrumental role in almost every communication system. For example, they can be used for  localization \cite{Pezeshki-Calderbank-08}, synchronization, channel estimation \cite{Spasojevic-Georghiades-01,Wang-Abdi-07}, and interference suppression/mitigation in multiuser systems \cite{Chen-Yeh-01,
Liu-Guan-Para-14,Liu-Guan-Chen-14}. An interesting sequence family is called ``complementary code" which was proposed by Tseng and Liu in \cite{Tseng-Liu-72}, where the aperiodic auto-correlation sum of all the constituent sequences equals zero at any nonzero time-shift. A special case of complementary code is Golay complementary pairs (GCPs) which were first found by Golay \cite{Golay-51} and each GCP consists of only two sequences.

Extensive  research attempts have been made concerning the constructions and applications of complementary codes and GCPs. A remarkable work was contributed by Davis and Jedwab \cite{Davis-Jedwab-99} who constructed polyphase GCPs  \cite{Golay-61} using the algebraic tool of generalized Boolean functions. Subsequently, representative constructions of complementary codes were proposed in  \cite{Paterson-00,Schmidt-06,ChenCY-16,ChenCY-17}. To obtain more flexible parameters,`` complementary code'' was extended to Z-complementary code (ZCC) in \cite{Fan-Yuan-Tu-07},
where the aperiodic auto-correlation sum is zero when the time-shift falls into a region around the in-phase position called zero-correlation zone (ZCZ).

Besides, Tseng and Liu studied in \cite{Tseng-Liu-72} ``mutually orthogonal complementary code set (MOCCS)'', which refers to a set of complementary codes with zero aperiodic cross-correlation sums between different complementary codes. Further developments of MOCCSs were reported in  \cite{ R.A-06, Wu-Spasojevi-08, Das-18, Das-SPL-18, Das-19}. However, MOCCS suffers from limited set size which is upper bounded by the number of constituent sequences (i.e., the flock size) in each code. To overcome this weakness, the concept of MOCCS was extended to Z-complementary code sets (ZCCSs) for higher set sizes \cite{Fan-Yuan-Tu-07}.

So far, the existing known constructions of ZCCSs only consider the front-end\footnote{with respect to the entire correlation window.} ZCZ of the aperiodic auto-correlation sums and aperiodic cross-correlation sums. Unlike the state-of-the-art works, we consider a specific subclass of ZCCSs named SZCCSs which exhibit both the front-end ZCZ and tail-end ZCZ properties. In practice, a front-end ZCZ and a tail-end ZCZ have particular interest for mitigating interference with small and large delays, respectively. We will show in Section IV that the properties of SZCCSs are useful for mitigating multipath interference and multi-antenna interference in GSM systems.

\subsection{Generalized Spatial Modulation (GSM)}

SM  is a multiple-input multiple-output (MIMO) technique which enjoys zero inter-channel interference over flat fading channels, low hardware complexity and low energy consumption. In SM, information is transmitted over two parts: 1) spatial dimension, coordinated by antenna indices, and 2) the conventional signal constellation of a modulation scheme \cite{Renzo-Haas-11}.
That is, information bits are conveyed through not only modulation symbols, but also the index of the active  transmit
antenna (TA) \cite{Naras-Rav-Cho-15}.
 Unlike SM system equipped with multiple TA elements but only a single radio-frequency (RF) chain, two or more TAs are activated in GSM at each symbol slot and the specific activated pattern itself conveys useful but implicit information \cite{Wang-Jia-12}. Therefore, GSM can achieve higher
spectral efficiency than SM systems while maintaining major advantages of SM \cite{Fu-Hou-10,Renzo-Haas-14}.  As a result, GSM strikes a flexible trade-off among spectral efficiency, cost of RF chains and energy efficiency by varying the number of RF chains\cite{He-Wang-18,Guo-Zhang-19}. These advantages make GSM a competitive candidate for the next generation wireless networks \cite{Wen-Zheng-19}.

There have been many iterative and non-iterative detector designs for GSM systems \cite{Yang-Xiao-16}. However, most of these detectors, e.g., the tree-search based detector \cite{Xiao-Yang-14} and the near-ML-detector \cite{Lin-Wu-15}, consider narrow-band scenarios where GSM symbols are transmitted over flat-fading channels. This is different from practical wireless channels which may exhibit frequency selective properties.  In recent years, various types of GSM detectors have been proposed for broadband GSM to combat the effects of inter-symbol interference imposed by  frequency-selective fading channels  \cite{Xiao-17,Wang-Liu-18,Zhao-Yang-18,Anand-Guan-20}. For example, a novel soft-decision feedback aided time-domain turbo equalizer based on the minimum mean-square error criterion \cite{Zhao-Yang-18}  was proposed for broadband GSM systems. In \cite{Sugiura-Hanzo-15}, a frequency-domain based turbo equalizer was developed for GSM under dispersive channels. In these works, perfect channel state information at the receiver was assumed. Recently, there appeared some literature on the channel estimators in GSM systems. For example, a channel estimation
scheme based on block pilot pattern and a interpolation method was proposed in \cite{Gong-18-GSM-OFDM}, which is designed for the GSM-orthogonal
frequency-division multiplexing system on high-speed railways. In \cite{Chu-19}, they proposed pilot-assisted and data-aided super-resolution MIMO
channel estimators for GSM-based millimeter-wave systems. Message-passing based blind channel-and-signal estimation and semi-blind channel-and-signal estimation algorithms were developed for massive MIMO systems with GSM \cite{Kuai-20}, which utilize the channel sparsity of the massive MIMO channel and the signal sparsity of GSM. However, to the best of our knowledge, all the existing methods for channel estimation in GSM systems are about the blind and semi-blind estimation algorithms, which usually have high computational complexity, and can not attain optimal estimation in theory.

  A major concern of this work is how to attain optimal estimation of channel state information in broadband GSM systems. A common means of obtaining channel state information is by sending properly designed preamble sequences at the transmitter, followed by correlating the known preamble sequences  at the receiver. It is noted that optimal channel estimation requires preamble sequences with zero nontrivial auto- and cross-correlations. However, the employment of conventional preamble sequences in GSM is not straightforward. Since only a few RF chains whose number is less than that of TAs are activated at each GSM time-slot and hence the transmit signal of GSM is sparse, it is hard to adopt those training sequences which are dense for traditional MIMO systems \cite{Yang-Wu-02,Fragouli-03,Fan-Mow-04}. On the other hand, when the number of RF chains is greater than one, training sequences designed for SM systems \cite{Liu-Yang-20} are also inapplicable.

\subsection{Novelty and Contributions of This Paper}

The main contributions and novelty of this paper are summarized as follows.
\begin{description}
  \item[1)] We introduce a novel class of ZCCSs called SZCCSs, each of which displaying zero tail-end auto/cross-correlation sums symmetrical to that of the front-end ones. More restrictive than the conventional ZCCSs, the design challenge of SZCCSs stems from the additional correlation properties associated to the tail-end ZCZ. We propose two systematic constructions of SZCCS with different set sizes and  sequence lengths using generalized Boolean functions. Especially, the SZCCSs generated from the first construction is optimal with power-of-two lengths, and the SZCCSs of non-power-of-two lengths constructed in the second construction have larger zero-correlation ratio than the former with specific parameters.
  \item[2)] We present a generic training framework for optimal GSM  training over frequency selective channels. The most distinctive feature of the proposed generic training framework (compared that for conventional MIMO systems) is that the training matrix should be sparse owing to the sparsity of the GSM transmit signals.  Based on this framework, we derive the optimal GSM training criteria under least square channel estimator and show that SZCCS plays a pivotal role in the design of optimal GSM training sequences. To the best of our knowledge, this paper is the first to consider the sequence design for the channel estimation in GSM systems. Numerical evaluations indicate that the proposed GSM training sequences lead to minimum channel estimation mean-square error and significantly outperform other classes of sequences (e.g., Zadoff-Chu sequences and binary random sequences) with different settings of activated transmit antennas and signal-to-noise ratios.
\end{description}

\subsection{Organization of This Paper}

The remainder of this paper is organized as follows. Section \ref{sec-pre} introduces some notations, ZCCSs and the mathematical tools used in the paper followed by a sketch of the basic principle of GSM. In Section III, we first present SZCCSs and show an upper bound of the size of SZCCSs. Then, we give two constructions of SZCCSs with different set sizes  based on generalized Boolean functions where one is optimal. In Section IV, first, we present a generic training framework for GSM training over frequency selective channels. Then, we derive the optimal GSM training criteria under LS channel estimator which implies that SZCCS plays an instrumental role in the design of optimal
GSM training sequences, and show some numerical evaluations of the proposed GSM training sequences. Finally, Section V concludes this paper with some remarks.

For readability, we summarize in Table \ref{table-short} all the acronyms which are used in this paper.
\begin{table}[htbp]
  \centering
 \caption{List of acronyms}\label{table-short}
  \begin{tabular}{|c|c|}
  \hline
 Acronyms&Descriptions\\  \hline
       CP & cyclic prefix \\
 CZCP & cross Z-complementary pair \\
  GCP&Golay complementary pair\\
 GSM & generalized spatial modulation \\
  LS &    least square\\
    MIMO &   multiple-input multiple-output \\
  MOCCS & mutually orthogonal complementary code set\\
 MSE&mean-square error\\
  RF &    radio-frequency\\
  SM & spatial modulation \\
  SZCCS &  symmetrical Z-complementary code set \\
   TA &   transmit antenna\\
  ZCC &  Z-complementary code \\
    ZCCS &   Z-complementary code set \\
  ZCZ &   zero-correlation zone \\
    ZP &   zero prefix \\
  \hline
\end{tabular}
\end{table}

\section{ Preliminaries}\label{sec-pre}

\subsection{Notations}

The following notations will be used throughout this paper.
\begin{itemize}
  \item $T^\tau(\mathbf{X})$ denotes the right-cyclic-shift of matrix $\mathbf{X}$ for $\tau$ (non-negative integer) positions over rows;
   \item $\mathbf{X}\|\mathbf{Y}$ denotes the concatenation of matrices $\mathbf{X}$ and $\mathbf{Y}$;
  \item $\mathbf{a}\|\mathbf{b}$ denotes the concatenation of sequences $\mathbf{a}$ and $\mathbf{b}$;
  \item $\mathbf{0}_{m\times n}$ denotes an all-0 matrix of order $m\times n$;
    \item $\lfloor j\rfloor_J$ denotes the modulo $J$ operation of integer $j$;
  \item Denote $\xi_q=exp\left(\frac{2\pi \sqrt{-1}}{q}\right)$, $q$ is a positive integer;
  \item Denote $\mathbb{Z}_q=\{0,1,...,q-1\}$ as the set of integers modulo $q$, where $q$ is a positive integer;
  \item $\mathcal{A}_q=\{w_q^0,w_q^1,...,w_q^{q-1}\}$ denotes the set over $q$ complex roots of unity.
\end{itemize}

Let $\mathbf{a}=(a(0),a(1),\cdots,a(L-1))$ and $\mathbf{b}=(b(0),b(1),\cdots,b(L-1))$ be two complex-valued sequences of length $L$.  The aperiodic cross-correlation function  between $\mathbf{a}$ and $\mathbf{b}$ at a time shift $u$ is defined by
\begin{eqnarray*}
 \rho_{\mathbf{a,b}}(u) =\left\{\begin{array}{ll}
\sum_{i=0}^{L-1-u}a(i)b^*(i+u), & 0\leq u\leq L-1; \\
\sum_{i=0}^{L-1+u}a(i-u)b^*(i), &-(L-1)\leq u\leq-1; \\
 0, & |u|\geq L.
   \end{array}\right.
\end{eqnarray*}
It is easy to verify that
\begin{eqnarray}\label{eq-u}
\rho_{\mathbf{a,b}}(u)=\rho^*_{\mathbf{b,a}}(-u).
\end{eqnarray}
If $\mathbf{a=b}$, $\rho_{\mathbf{a,b}}$ is called the aperiodic auto-correlation function, denoted by $\rho_\mathbf{a}$ for simplicity.

Also, denote by $\phi_{\mathbf{a,b}}(u)$ the periodic cross-correlation between $\mathbf{a}$ and $\mathbf{b}$, i.e., $$\phi_{\mathbf{a,b}}(u)=\sum_{i=0}^{L-1}a(i)b^*(\lfloor i+u\rfloor_L).$$

In particular, let $\mathcal{A}=\{\mathbf{a}_1,\mathbf{a}_2,...,\mathbf{a}_M\}$ and $\mathcal{B}=\{\mathbf{b}_1,\mathbf{b}_2,...,\mathbf{b}_M\}$ be two sequence sets of size $M$ and length $L$.  The aperiodic cross-correlation function between $\mathcal{A}$ and $\mathcal{B}$ at a time shift $u$ is defined by
$
 C_{\mathcal{A,B}}(u) =\sum_{m=1}^M\rho_{\mathbf{a}_m,\mathbf{b}_m}(u).
$
Similarly, if $\mathcal{A=B}$, $C_{\mathcal{A,B}}$ is called the aperiodic auto-correlation function of sequence set $\mathcal{A}$, denoted by  $C_\mathcal{A}$ for simplicity.
\subsection{Generalized Boolean Functions}

Let $q$ be a positive integer,
for $\mathbf{x}=(x_1,x_2,\cdots,x_m)\in\mathbb{Z}_2^m$, a generalized Boolean function $f(\mathbf{x})$ is defined as a mapping $f$ from $\{0,1\}^m$ to $\mathbb{Z}_q$.
Given $f(\mathbf{x})$, define
\begin{eqnarray}\label{eq-GBF}
  \mathbf{f}=(f(0),f(1),...,
  f(2^{m}-1)),
\end{eqnarray}
where
$f(i)=f(i_1,i_2,\cdots,i_{m})$, and $(i_1,i_2,\cdots,i_{m})$ is the binary representation of $i=\sum_{k=1}^mi_k2^{k-1}$ with $i_{m}$ denoting the most significant bit.

In this paper, we consider truncated versions of the sequence $\mathbf{f}$ of Eq. (\ref{eq-GBF}). Specifically,  let  $\mathbf{f}^{(L)}$ be a sequence of length $L$ obtained from $\mathbf{f}$ by ignoring the last $2^m-L$ elements of the sequence $\mathbf{f}$. That is, $\mathbf{f}^{(L)}=(f(0),f(1),
\cdots,f(L-1))$ is a sequence of length $L$.  Let $\xi_q=\exp(2\pi\sqrt{-1}/q)$ be a primitive $q$-th complex root of unity. One can naturally associate a complex-valued sequence $\psi(\mathbf{f}^{(L)})$ of length $L$ with $\mathbf{f}^{(L)}$ as
\begin{eqnarray}\label{eq-fL}
  \psi(\mathbf{f}^{(L)}) &:=&(\xi_q^{f(0)},\xi_q^{f(1)},
  \cdots,\xi_q^{f(L-1)}).
\end{eqnarray}
From now on, whenever the context is clear, we ignore the superscript of $\mathbf{f}^{(L)}$ unless the sequence length is specified.

\subsection{Introduction to ZCCS}
\begin{Definition}
Let ${\mathcal{A}}=\{\mathbf{a}_{m}\}_{m=1}^{M}$ be a set of $M$ complex-valued sequences of length $L$.  It is said to be a (aperiodic) $(M,L,Z)$-ZCC of size $M$ if $C_{\mathcal{A}}(u)=0$ for any  $0<|u|\leq Z$ where $Z$ is a positive integer with $1\leq Z\leq L-1$. In particular, when $Z=L-1$, the set is called a (aperiodic) complementary code, and when $M=2$ it is called a (aperiodic) GCP.
\end{Definition}

The following lemma shows a construction of GCPs, which will be used in the sequel.
\begin{Lemma}[\textbf{Corollary 11} of \cite{Paterson-00}, \textbf{Theorem 3.3} of \cite{R.A-08}] \label{lem-GCP}
Let $q$ be an even integer and $m$ be a positive integer. Let
\begin{align*}
a(\mathbf{x}) =& \frac{q}{2}\sum_{k=1}^{m-1}x_{\pi(k)}x_{\pi(k+1)}+\sum_{k=1}^mc_kx_k+c, \\
b(\mathbf{x}) =& a(\mathbf{x})+ \frac{q}{2}x_{\pi(1)},\\
c(\mathbf{x}) =& a(\mathbf{x})+ \frac{q}{2}x_{m},
\end{align*}
where $\pi$ is a permutation of $\{1,2,\cdots,m\}$ and $\mathbf{x}\in\mathbb{Z}_2^m,c_k,c\in\mathbb{Z}_q$. Then $(\psi(\mathbf{a}),\psi(\mathbf{b}))$ and $(\psi(\mathbf{a}),\psi(\mathbf{c}))$ are GCPs of length $2^m$.
\end{Lemma}

\begin{Definition}\label{def-MOCCS}
Let ${\mathcal{S}}=\{{\mathcal{S}}_1,{\mathcal{S}}_2,\cdots,{\mathcal{S}}_K\}$, where each ${\mathcal{S}}_k=\{\mathbf{s}_{k,1},\mathbf{s}_{k,2},...,\mathbf{s}_{k,M}\}$ $(1\leq k\leq K)$ be a ZCC consisting of $M$ length-$L$ sequences. ${\mathcal{S}}$ is called a $(K,M,L,Z)$-ZCCS if
$
C_{\mathcal{S}_i,\mathcal{S}_k}(u) =0, ~~\forall~~1\le i\ne k\le K  ~\textrm{and}~ |u|\leq Z.
$
In particular, when $Z=L$, it is called a $(K,M,L)$-MOCCS.
\end{Definition}

The following lemma is about the upper bound of the size of ZCCS, which was proposed in \cite{Fan-Yuan-Tu-07}.
\begin{Lemma}\label{lem-bound}(\cite{Fan-Yuan-Tu-07}, \cite{Liu-Guan-Ng-11})
Any unimodular $(K,M,L,Z)$-ZCCS $\mathcal{S}=\{\mathcal{S}_1,\mathcal{S}_2,...,\mathcal{S}_K\}$ satisfies $K\leq \left\lfloor \frac{ML}{Z+1}\right\rfloor,$ where $\mathcal{S}_k=\{\mathbf{s}_{k,1},\mathbf{s}_{k,2},\cdots,\mathbf{s}_{k,M}\}$ and $\mathbf{s}_{k,m}=\left(s_{k,m}(0), \right.$ $\left. s_{k,m}(1),\cdots,s_{k,m}(L-1)\right)$ $(1\leq k\leq K,1\leq m\leq M)$. In particular, when $K= \left\lfloor \frac{ML}{Z+1}\right\rfloor$, $\mathcal{S}$ is called  an optimal $(K,M,L,Z)$-ZCCS.
\end{Lemma}

\subsection{Introduction to GSM}
We consider a single-carrier GSM (SC-GSM) system with $N_t$ TA elements, $N_r$ receive antennas and $N_{\text{active}}$ transmit RF chains over frequency-selective channels. Moreover, we consider a QAM/PSK modulation with constellation size of $\mathcal{M}_{\text{GSM}}$. An $N_{\text{active}}\times N_t$ switch connects the RF chains to the TAs. In a given channel use, each user selects $N_{\text{active}}$ in $N_t$ TAs, and transmit $N_{\text{active}}$ symbols from a QAM/PSK modulation alphabet $\mathbb{A}$ on the selected antennas. The remaining $N_t-N_{\text{active}}$ antennas remain silent. Fig. \ref{GSM} shows the GSM transmitter at the user terminal. Over each time-slot $k$, there are $\left\lfloor \log_2{N_t\choose{N_{\text{active}}}}\right\rfloor+\left\lfloor \log_2|\mathbb{A}|\right\rfloor$ bits, denoted by $\mathbf{b}$, conveyed by a GSM transmitter. Specially,  the first $\left\lfloor \log_2{N_t\choose{N_{\text{active}}}}\right\rfloor$ bits, denoted by $\mathbf{b}_1$, are used to active the $i_1\text{th},i_2\text{th},...,i_{N_{\text{active}}}\textit{th}$ TAs through a antenna activation pattern selector, which is determined by the mapping between information bits and antenna activation patterns. The table in Fig. \ref{GSM} gives an example of that mapping for $N_t=5$ and $N_{\text{active}}=2$. Suppose the selected antenna activation pattern is denoted by $\mathbf{s}$, where ``1'' in $\mathbf{s}$  indicates that the antenna corresponding
to that coordinate is active and silent otherwise. The last $\left\lfloor \log_2|\mathbb{A}|\right\rfloor$ bits, denoted by $\mathbf{b}_2$, are used to select $N_{\text{active}}$ ``constellation symbols'' $\mathbf{B}_1,\mathbf{B}_2,\cdots,\mathbf{B}_{N_{\text{active}}}$ through a modulation symbol mapper, conveyed
through the $N_{\text{active}}$ activated antennas after adding zero
prefix (ZP) or cyclic prefix (CP) to combat dispersive GSM channels, respectively.
 Details of GSM transmit principle can be found in \cite{Naras-Rav-Cho-15}.
\begin{figure*}[htbp]
  \centering
  \includegraphics[width=6in]{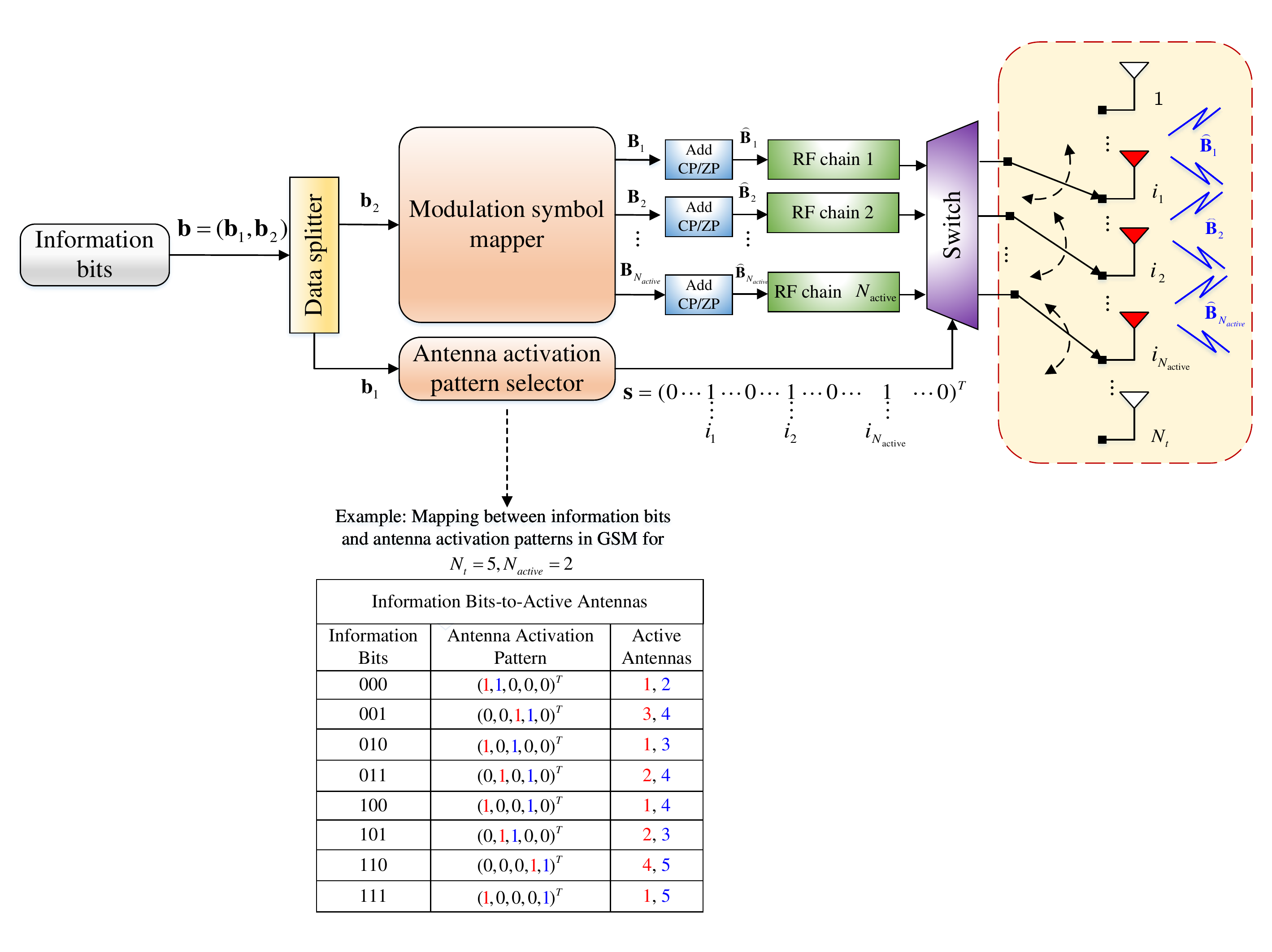}\\
  \caption{The structure of a GSM transmitter, where $\left\{i_1,i_2,\cdots,i_{N_{\text{active}}}\right\}$ denotes the index set of activated TAs.}\label{GSM}
\end{figure*}
\begin{Example}
Consider an SC-GSM system with $N_{\text{active}}=2$ RF chains and $N_t=5$ TAs using BPSK modulation $\mathcal{M}_{\text{GSM}}=2$. Specifically, two out of the five transmit  antennas are activated at each time-slot using $\left\lfloor \log_2{5\choose{2}}\right\rfloor=3$ information bits. In total,  there are only $2^{\left\lfloor log_2{5\choose{2}}\right\rfloor}=8$ activation patterns for signaling. The mapping of information bits to an activation pattern in GSM refers to the table in Fig. \ref{GSM}.
For illustration purpose, we consider natural mapping for BPSK modulation over each activated TA. Suppose each SC-GSM block constitutes 4 GSM symbols. Suppose further these symbols correspond to 16 message bits $(1001111011000100)$, with symbols $\mathbf{b}(1)=(1001)$, $\mathbf{b}(2)=(1110)$, $\mathbf{b}(3)=(1100)$ and $\mathbf{b}(4)=(0100)$. Taking the first GSM symbol for example, we have $\mathbf{b}(1)=(\mathbf{b}_1(1),\mathbf{b}_2(1))$ where $\mathbf{b}_1(1)=(100)$ selects the antenna activation pattern $\mathbf{s}=(1,0,0,1,0)^T$, and $\mathbf{b}_{2}(1)=(1)$ which implies that the modulated symbols $\mathbf{B}_1=\mathbf{B}_2=\xi_2^1=-1$. This means that during the first time-slot, the first and forth TAs are activated for the sending of BPSK symbol $-1$. Then, the first GSM symbol can be written as $\mathbf{d}_1=(-1,0,0,-1,0)^T$. The entire SC-GSM block can be expressed by the sparse matrix:
$$
 (\mathbf{d}_1,\mathbf{d}_2,\mathbf{d}_3,\mathbf{d}_4)=
 \left(
  \begin{array}{cccc}
  -1&1 &0 & 1   \\
  0& 0 &0 & 0 \\
  0& 0 & 0 & 1 \\
 -1 & 0 & 1 & 0 \\
 0 & 1 & 1 & 0 \\
  \end{array}
  \right).
$$
Note that  if SM is used instead with the same modulation
order, the number of transmit antennas must be increased to
eight to maintain the same spectral efficiency.
\end{Example}

\section{Symmetrical $Z$-complementary code set (SZCCS): Properties And Constructions}

This section presents the properties and  constructions of SZCCSs. The definition and some properties of SZCCSs are as follows.
\begin{Definition}\label{def-SZCCS}
For a positive integer $Z$, $\mathcal{S}=\{\mathcal{S}_1,\mathcal{S}_2,...,\mathcal{S}_K\}$ is called a $(K,M,L,$ $Z)$-SZCCS where $\mathcal{S}_k=\{\mathbf{s}_{k,1},\mathbf{s}_{k,2},...,\mathbf{s}_{k,M}\}$ and $\mathbf{s}_{k,m}=(s_{k,m}(0),s_{k,m}(1),...,$ $s_{k,m}(L-1))$ $(1\leq k\leq K,1\leq  m\leq M)$, if $\mathcal{S}$ satisfies the following conditions:
\begin{description}
  \item[C1:]$C_{\mathcal{S}_k}(u)=0$, for $|u|\in \mathcal{T}_1\bigcup \mathcal{T}_2$;
  \item[C2:]$C_{\mathcal{S}_k,\mathcal{S}_j}(u)=0$, for $k\ne j$ and $|u|\in \mathcal{T}_1\bigcup \mathcal{T}_2\bigcup\{0\}$;
\end{description}
where $\mathcal{T}_1=\{1,2,...,Z\}$ and $\mathcal{T}_2=\{L-Z,L+1-Z,...,$ $L-1\}$.
\end{Definition}

\begin{Remark}
  Note that SZCCS considers the cross-correlation function sum between the different ZCCs, whereas, CZCP proposed in \cite{Liu-Yang-20} is characterized by the cross-correlation function sum of the sequences in a same ZCP.
\end{Remark}

\begin{figure*}[htbp]
  \centering
  \includegraphics[width=5.5IN]{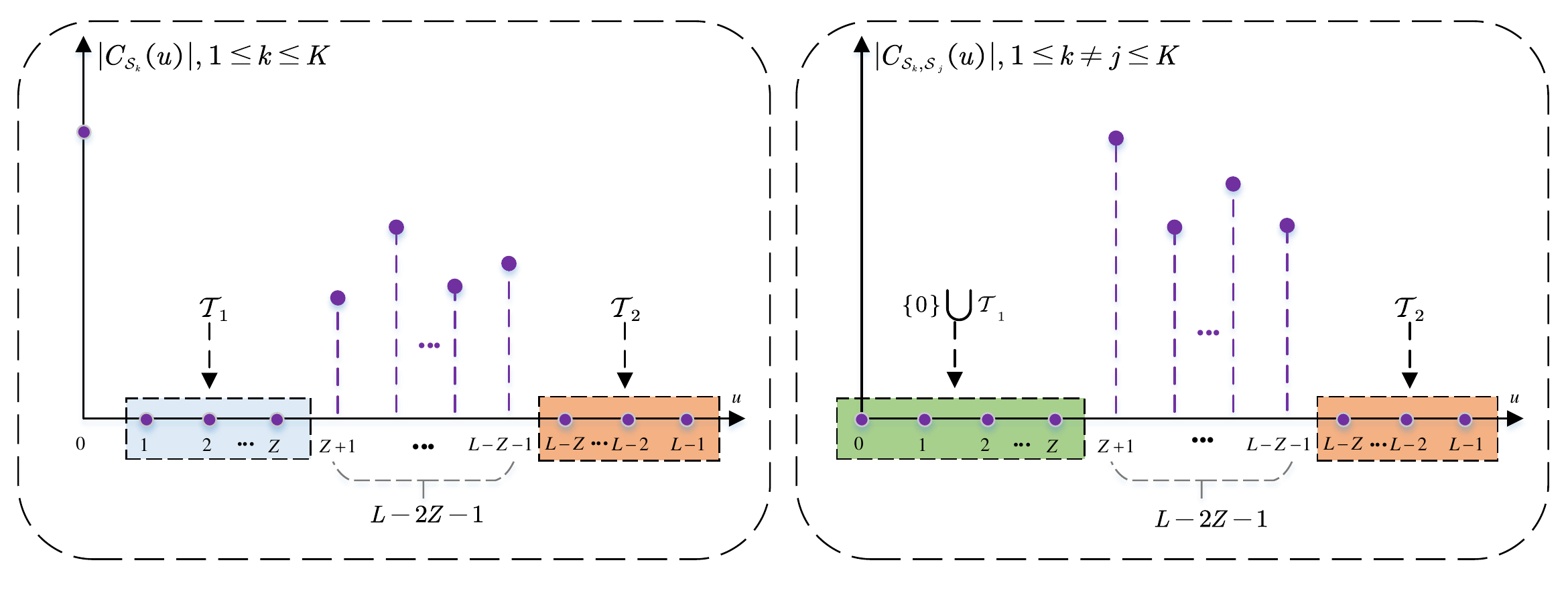}\\
  \caption{The correlation properties of $(K,M,L,Z)$-SZCCS}\label{fig-def}
\end{figure*}

Fig. \ref{fig-def} illustrates the correlation properties of $(K,M,L,Z)$-SZCCS which shows that there are ZCZs at both the front and tail of auto- or cross- correlation sums.  Note that when $Z\geq\frac{L-1}{2}$, a $(K,M,L,Z)$-SZCCS reduces to a $(K,M,L)$-MOCCS defined in \textbf{Definition \ref{def-MOCCS}}. Hence,
the following theorem is straightforward from \textbf{Lemma \ref{lem-bound}}.

\begin{Theorem}\label{thm-bound}
Any unimodular $(K,M,L,Z)$-SZCCS $\mathcal{S}=\{\mathcal{S}_1,$ $\mathcal{S}_2,...,\mathcal{S}_K\}$ satisfies $K\leq \left\lfloor \frac{ML}{Z+1}\right\rfloor,$ where $\mathcal{S}_k=\{\mathbf{s}_{k,1},\mathbf{s}_{k,2},$ $\cdots,\mathbf{s}_{k,M}\}$ and $\mathbf{s}_{k,m}=(s_{k,m}(0), s_{k,m}(1),$ $\cdots,s_{k,m}(L-1))$ $(l\leq k\leq K,1\leq m\leq M)$. In particular, when $K= \left\lfloor \frac{ML}{Z+1}\right\rfloor$, $\mathcal{S}$ is called  an optimal $(K,M,L,Z)$-SZCCS.
\end{Theorem}
\begin{Corollary}
For any $(K,M,L,Z)$-SZCCS $\mathcal{S}=\{\mathcal{S}_1,\mathcal{S}_2,\cdots,$ $\mathcal{S}_K\}$ where $\mathcal{S}_k=\{\mathbf{s}_{k,1},\mathbf{s}_{k,2},...,\mathbf{s}_{k,M}\}$ and $\mathbf{s}_{k,m}=(s_{k,m}(0),$ $s_{k,m}(1),...,s_{k,m}(L-1))$ $(1\leq k\leq K,1\leq  m\leq M)$, it satisfies the following properties:
\begin{description}
  \item[P1:]$\mathcal{\tilde{S}}=\{c_1\mathcal{S}_1,c_2\mathcal{S}_2,...,c_K\mathcal{S}_K\}$
      is also a $(K,M,L,Z)$-SZCCS, where $c_k\in\mathbb{C}$ $(1\leq k\leq K)$ is nonzero and $c_k\mathcal{S}_k=\{c_k\mathbf{s}_{k,1},c_k\mathbf{s}_{k,2},...,c_k\mathbf{s}_{k,M}\}$;
  \item[P2:] $\overleftarrow{\mathcal{S}}=\{\overleftarrow{\mathcal{S}_1},\overleftarrow{\mathcal{S}_2},...,\overleftarrow{\mathcal{S}_K}\}$ is also a $(K,M,L,Z)$-SZCCS, where $\overleftarrow{\mathcal{S}_k}=\{\overleftarrow{\mathbf{s}_{k,1}},\overleftarrow{\mathbf{s}_{k,2}},$ $...,\overleftarrow{\mathbf{s}_{k,M}}\}$ and $\overleftarrow{\mathbf{s}_{k,m}}$ $=\left(s_{k,m}(L-1),s_{k,m}(L-2),...,s_{k,m}(0)\right)~(1\leq m$ $\leq M)$.
\end{description}
Note that if $\mathcal{S}$ is optimal, $\mathcal{\tilde{S}}$ and $\overleftarrow{\mathcal{S}}$ are also optimal.
\end{Corollary}
\begin{proof}
The proof is straightforward according the definitions of aperiodic auto-correlation, aperiodic cross-correlation and SZCCS, so we omit it here.
\end{proof}
\subsection{Proposed construction of optimal SZCCS}
\begin{Theorem}\label{thm-optimalcon}
Let $q$ be an even integer, $m\geq4$ and  $L=2^m$, and let
\begin{align*}
 f(\mathbf{x})=& \frac{q}{2}\sum_{s=1}^{m-2}x_{\pi(s)}
  x_{\pi(s+1)}+\sum_{s=1}^{m}\mu_{s}
  x_{s}+\mu,\\
  a_k(\mathbf{x})=& f(\mathbf{x})+d_1^k\left(\frac{q}{2} x_{\pi(m-2)}+\frac{q}{2}x_{\pi(m-2)}x_{\pi(m)}\right)\end{align*}
     \begin{align*}
  &+d_2^k\left(\frac{q}{2} x_{\pi(m-1)}+\frac{q}{2}x_{\pi(m-2)}x_{\pi(m)}\right)\\
  &+d_3^k\left(\frac{q}{2} x_{\pi(m)}+\frac{q}{2}x_{\pi(m-1)}x_{\pi(m)}\right)\\
  &+d_4^k\left(\frac{q}{2} x_{\pi(m-1)}+\frac{q}{2}x_{\pi(m-2)}x_{\pi(m)}\right.\\&\left.+\frac{q}{2}x_{\pi(m-1)}x_{\pi(m)}\right),\\
     b_k(\mathbf{x})=& a_k(\mathbf{x})
  +\frac{q}{2}x_{\pi(1)},
   \end{align*}
where $\mathbf{\underline{x}}\in \mathbb{Z}_{2}^{m},~\mu,\mu_s$ are any given elements in $\mathbb{Z}_q$, $\pi$ is a permutation of the symbols $\{1,2,\cdots,m\}$ with $\{\pi(m-1),\pi(m)\}=\{m-1,m\}$ and $\left(\mathbf{D}_k=(d_1^k,d_2^k,d_3^k,d_4^k)\right)_{k=1}^8=((0,0,0,0),(1,0,1,0),(1,1,0,0),(0,1,1,0),
(0,0,0,1),(1,0,$ $1,1),(1,1,0,1),(0,1,1,1))$.
Then, the set $\mathcal{S}=\left\{\mathcal{S}_k=\left\{\psi(\mathbf{a}_k),\psi(\mathbf{b}_k)\right\}:
k\in\{1,2,\cdots,8\}\right\}$
is an optimal $(8,2,2^{m},2^{m-2}-1)$-SZCCS.
\end{Theorem}
\emph{Proof of} \textbf{Theorem \ref{thm-optimalcon}}: See Appendix A.

\begin{Example}\label{eg-optimal}
 For $q=2$ and $m=4$, let $\pi$ be a permutation of $\{1,2,3,4\}$ with $\pi(i)=i$ for $1\leq i\leq4$, and $f(\mathbf{x})=\sum\limits_{k=1}^2x_kx_{k+1}$.
Then, $\mathcal{S}=\left\{\mathcal{S}_k=\left\{\psi(\mathbf{a}_k),\psi(\mathbf{b}_k)\right\}:
k\in\{1,2,\cdots,8\}\right\}$
is an optimal $(8,2,16,3)$-SZCCS.
$ \left(|C_{\mathcal{S}_i,\mathcal{S}_j}(u)|\right)_{u=0}^{15}=
 \left(|\rho_{\psi(\mathbf{a}_i),\psi(\mathbf{a}_j)}(u)+\rho_{\psi(\mathbf{b}_i),
  \psi(\mathbf{b}_j)}(u)|\right)_{u=0}^{15}~~(1\leq i,j\leq8)$ can be found in Table \ref{table-eg-optimal}.
\begin{sidewaystable}
  \centering
\caption{$\mathbf{C}=\left(|C_{\mathcal{S}_i,\mathcal{S}_j}(u)|\right)_{u=0}^{15}$ in Example \ref{eg-optimal}}\label{table-eg-optimal}
  {\footnotesize
\begin{tabular}{|c|c|c|c|c|c|c|c|c|}
  \hline
  \diagbox[width=3em,trim=l]{$i$}{$\mathbf{C}$}
  {$j$} & 1  &  2&3& 4&5& 6 & 7 &8\\ \hline
  1&   $(32,\mathbf{0}_7,16,\mathbf{0}_7)$ & $(\mathbf{0}_4,24,\mathbf{0}_7,8,\mathbf{0}_3)$ & $(\mathbf{0}_4,-8,\mathbf{0}_{11})$  & $(\mathbf{0}_{12},8,\mathbf{0}_{3})$ &$\left(\mathbf{0}_{12},8,\mathbf{0}_{3}\right)$ & $(\mathbf{0}_4,8,\mathbf{0}_{11})$ &$(\mathbf{0}_4,8,\mathbf{0}_{11})$&$\left(\mathbf{0}_8,16,\mathbf{0}_{7}\right)$\\ \hline
  2&$(\mathbf{0}_4,24,\mathbf{0}_7,8,\mathbf{0}_3)$ & $(32,\mathbf{0}_7,16,\mathbf{0}_7)$&$\left(\mathbf{0}_{12},8,\mathbf{0}_{3}\right)$  & $ \left(\mathbf{0}_{4},8,\mathbf{0}_{11}\right)$  & $\left(\mathbf{0}_{4},8,\mathbf{0}_{11}\right)$&$(\mathbf{0}_{12},8,\mathbf{0}_{3})$
  &$\left(\mathbf{0}_8,16,\mathbf{0}_{7}\right)$&$(\mathbf{0}_{12},8,\mathbf{0}_{3})$\\ \hline
  3&   $\left(\mathbf{0}_{12},8,\mathbf{0}_3\right)$ & $(\mathbf{0}_4,8,\mathbf{0}_{11})$ & $(32,\mathbf{0}_3,8,\mathbf{0}_7,8,\mathbf{0}_3)$  &  $\left(\mathbf{0}_{8},16,\mathbf{0}_{7}\right)$ &$(\mathbf{0}_{4},16,\mathbf{0}_{3},16,\mathbf{0}_{7})$ & $(\mathbf{0}_4,8,\mathbf{0}_{7},8,\mathbf{0}_3)$ &$\left(\mathbf{0}_4,8,\mathbf{0}_{11}\right)$&$(\mathbf{0}_4,16,\mathbf{0}_{7},8,\mathbf{0}_3)$\\ \hline
  4& $\left(\mathbf{0}_{4},8,\mathbf{0}_{11}\right)$ &$(\mathbf{0}_{12},8,\mathbf{0}_{3})$&$\left(\mathbf{0}_{8},16,\mathbf{0}_{7}\right)$& $(32,\mathbf{0}_3,8,\mathbf{0}_7,8,\mathbf{0}_3)$  & $(\mathbf{0}_{4},8,\mathbf{0}_{11})$&$(\mathbf{0}_4,16,\mathbf{0}_{3},16,\mathbf{0}_7)$&
  $(\mathbf{0}_4,16,\mathbf{0}_{7},8,\mathbf{0}_3)$&$(\mathbf{0}_4,-8,\mathbf{0}_{11})$\\ \hline
  5& $(\mathbf{0}_{4},8,\mathbf{0}_{11})$ & $(\mathbf{0}_{12},8,\mathbf{0}_{3})$ & $(\mathbf{0}_{4},16,\mathbf{0}_{3},16,\mathbf{0}_{7})$  &  $\left(\mathbf{0}_{4},8,\mathbf{0}_{11}\right)$ &$(32,\mathbf{0}_3,8,\mathbf{0}_7,8,\mathbf{0}_3)$ &$(\mathbf{0}_8,16,\mathbf{0}_7)$ &$(\mathbf{0}_4,16,\mathbf{0}_{7},8,\mathbf{0}_3)$&$\left(\mathbf{0}_4,8,\mathbf{0}_{11}\right)$\\ \hline
  6&  $\left(\mathbf{0}_{12},8,\mathbf{0}_{3}\right)$  & $\left(\mathbf{0}_{4},8,\mathbf{0}_{11}\right)$& $(\mathbf{0}_{4},8,\mathbf{0}_{7},8,\mathbf{0}_{3})$  & $(\mathbf{0}_{4},16,\mathbf{0}_{3},16,\mathbf{0}_{7})$ & $(\mathbf{0}_8,16,\mathbf{0}_7)$&
  $(32,\mathbf{0}_3,8,\mathbf{0}_7,8,\mathbf{0}_3)$&$(\mathbf{0}_4,8,\mathbf{0}_{11})$&
  $(\mathbf{0}_4,16,\mathbf{0}_{7},8,\mathbf{0}_3)$\\ \hline
   7&$(\mathbf{0}_{4},8,\mathbf{0}_{7},8,\mathbf{0}_{3})$ &$(\mathbf{0}_{8},16,\mathbf{0}_{7})$&$(\mathbf{0}_{4},16,\mathbf{0}_{7},8,\mathbf{0}_{3})$ &  $\left(\mathbf{0}_{4},8,\mathbf{0}_{11}\right)$ &$\left(\mathbf{0}_4,8,\mathbf{0}_{11}\right)$&$(\mathbf{0}_4,16,\mathbf{0}_{7},8,\mathbf{0}_3)$&
   $(32,\mathbf{0}_7,16,\mathbf{0}_7)$&
   $(\mathbf{0}_4,8,\mathbf{0}_{7},8,\mathbf{0}_3)$\\ \hline
  8& $(\mathbf{0}_{8},16,\mathbf{0}_{7})$ &$(\mathbf{0}_{4},8,\mathbf{0}_{7},8,\mathbf{0}_{3})$& $(\mathbf{0}_{4},8,\mathbf{0}_{11})$  &$(\mathbf{0}_{4},16,\mathbf{0}_{7},8,\mathbf{0}_{3})$ & $(\mathbf{0}_4,16,\mathbf{0}_{7},8,\mathbf{0}_{3})$&$ (\mathbf{0}_4,8,\mathbf{0}_{11})$&$(\mathbf{0}_4,8,\mathbf{0}_{7},8,\mathbf{0}_3)$&
  $(32,\mathbf{0}_7,16,\mathbf{0}_7)$\\ \hline
\end{tabular}}
\end{sidewaystable}
\end{Example}

\begin{Remark}
The significance of \textbf{Theorem 1} is reflected in the following aspects.

\texttt{For sequence design:}
\begin{itemize}
    \item It provides an infinite class of optimal SZCCSs, which means that it can generate optimal SZCCSs for infinite groups of parameters.
    \item This construction is interesting since it is based on two facts we observed.
        \begin{enumerate}
          \item For the aperiodic correlation sum of sequence pairs in \textbf{Lemma 1}, the width of ZCZ $Z$  can be controlled by the permutation $\pi$, which means that if $\{\pi(1),\pi(2),\cdots,\pi(v)\}=\{1,2,\cdots,v\}$ for a positive $v$ where $1\leq v<m-1$, then, we have $Z=2^v-1$.
          \item The orthogonality of different sequence pairs in \textbf{Lemma 1} can be guaranteed by some offsets composed of some linear functions.
        \end{enumerate}
    \item The $(8,2,2^m,2^{m-2})$-SZCCS constructed in \textbf{Theorem 1}  achieves the theoretical upper bound on
the size of the SZCCS proposed in \textbf{Theorem 1}, which implies that for any $K$, if there exists a $(K,2,2^m,2^{m-2}-1)$-SZCCS, then, we have $K\leq8$. Furthermore, a $(K,2,2^m,2^{m-2}-1)$-SZCCS can be obtained by selecting any $K$ codes from an $(8,2,2^m,2^{m-2})$-SZCCS.
  \end{itemize}

\texttt{For application in GSM systems:}
  \begin{itemize}
    \item For a $(K,M,L,Z)$-SZCCS, $K$ represents the maximum number of RF chains that can be supported, $Z+1$ represents the number of multi-paths, $M(L+Z)$ represents the length of training sequences under the proposed training framework in the sequel. Hence, for fixed number of multi-path and training sequence length, the bigger the $K$, the better. Note that $(8,2,2^m,$ $2^{m-2})$-SZCCS constructed in \textbf{Theorem 1} is optimal, so $K=8$ is the best for the GSM systems under the proposed training framework in the sequel with $2^{m-2}$ multi-paths and training sequence length $2(2^m+2^{m-2}-1)$.
  \end{itemize}
\end{Remark}

However, the sequence length of the proposed optimal SZCCS is limited to a power-of-two. The following subsection shows a construction of SZCCSs, which are not optimal but with non-power-of-two sequence length.

\subsection{Proposed construction of SZCCS with non-power-of-two lengths}
  \begin{Theorem}\label{thm-SZCCS-non1}
Let $q$ be an even integer, $m\geq3,~1\leq v<m-1$ and  $L=2^{m-1}+2^v$, and let
\begin{align*}
 g(\mathbf{x})=& \frac{q}{2}\sum_{s=1}^{m-2}x_{\pi(s)}
  x_{\pi(s+1)}+\sum_{s=1}^{m-1}\lambda_{s}
  x_{\pi(s)}x_m+\sum_{s=1}^{m}\mu_{s}
  x_{s}+\mu, \\
  a_k(\mathbf{x})=&g(\mathbf{x})
  +\frac{q}{2}(k-1)\left(x_{m-1}+x_mx_{\pi(v)}\right),\\
 b_k(\mathbf{x})=&a_k(\mathbf{x})
 +\frac{q}{2}x_{\pi(1)},\\
  c_k(\mathbf{x})=&a_k(\mathbf{x})+\frac{q}{2}x_{m},\\
  d_{k}(\mathbf{x}) =& c_{k}(\mathbf{x})+\frac{q}{2}
 x_{\pi(1)},
\end{align*}
where $\mathbf{x}\in \mathbb{Z}_{2}^{m},~k\in\{1,2\},~\lambda_s,\mu,\mu_s$ are any given elements in $\mathbb{Z}_q$, and $\pi$ is a permutation of the symbols $\{1,2,\cdots,m\}$ with $\{\pi(1),\pi(2),\cdots,\pi(v)\}=\{1,2,\cdots,v\}$ and $\pi(m)=m$.
Then,  $\mathcal{S}=\left\{\mathcal{S}_k=\{\psi(\mathbf{a}_k),\psi(\mathbf{b}_k)\}:k=1,2\right\}$ and $\mathcal{S}'=\left\{\mathcal{S}'_k=\{\psi(\mathbf{c}_k),\psi(\mathbf{d}_k)\}:k=1,2\right\}$  are $(2,2,2^{m-1}+2^v,2^v-1)$-SZCCSs.
\end{Theorem}
\begin{proof}
  The proof of \textbf{Theorem \ref{thm-SZCCS-non1}} is similar with that of \textbf{Theorem \ref{thm-optimalcon}}, so we omit it here.
\end{proof}

\begin{Remark}
When $v=m-2$ in \textbf{Theorem \ref{thm-SZCCS-non1}}, one can generate $(2,2,2^{m-1}+2^{m-2},2^{m-2}-1)$-SZCCSs. However, to obtain SZCCSs with same $Z=2^{m-2}-1$ using \textbf{Theorem \ref{thm-optimalcon}}, the sequence length needs to be $L=2^m$, which is larger than that of SZCCSs generated from \textbf{Theorem \ref{thm-SZCCS-non1}}.
\end{Remark}

\section{Optimal Training Design Using SZCCSs for Broadband GSM Systems}

In this section, based on the generic training-based SC-MIMO transmission structure as shown in \cite{Liu-Yang-20}, we first present a GSM training framework using sparse matrices and derive the correlation properties of the row sequences of such a sparse matrix. Then, we show that the SZCCSs (proposed in Section III) can be utilized as a key component in optimal GSM training design.

We consider the training setting with a length-$\lambda$ CP and $N_t$ TAs over quasi-static frequency-selective channel. We denote the channel impulse response (CIR) from the $n$-th ($1\leq n\leq N_t$) transmit antenna to the receiver as $\mathbf{h}_n=(h_{n,0},h_{n,1},...,h_{n,\lambda})^T$ where $h_{n,l}$ $(0\leq l\leq\lambda)$ is the channel coefficient of the $l$-th path. Note that there is a training sequence followed by data payload in each block at a TA and CP is placed at the front of the training sequence \cite{Liu-Yang-20}. Let $\mathbf{x}_n=(x_{n,0},x_{n,1},...,x_{n,L-1})$ be the training sequence transmitted over the $n$-th TA. We assume that all the training sequences have identical energy of $E$.
Then, the minimum MSE  under LS channel estimator is achieved if and only if
\begin{align}\label{eq-MSE}
\phi_{\mathbf{x}_i,\mathbf{x}_j}(u)&=\left\{\begin{array}{ll}
                                             E, & \hbox{if}~i=j,~u=0, \\
                                             0, & \hbox{if}~i\ne j,~0\leq u\le\lambda,\\
                                             0, & \hbox{if}~i=j,~1\leq u\leq\lambda, \end{array}\right.
\end{align}
with
\begin{equation}\label{MSE}
 \hbox{minimum~MSE}=\frac{\sigma^2_w}{E}
\end{equation}
where $\sigma^2_w$ is the variance of the white complex Gaussian noise \cite{Liu-Yang-20}.
\begin{Definition}
For training sequences $\{\mathbf{x}_{n}\}_{n=1}^{N_t}$ where the sequences have identical energy of $E$, they are called optimal training sequences of  SC-MIMO systems under LS channel estimator if and only  they satisfy Eq. (\ref{eq-MSE}).
\end{Definition}
\subsection{Proposed Training Framework For Broadband GSM Systems}\label{sub-framework}

We define the training matrix $\Omega$ as $$
 \Omega =\left(
           \begin{array}{c}
             \mathbf{x}_1 \\
             \mathbf{x}_2 \\
             \vdots \\
             \mathbf{x}_{N_t} \\
           \end{array}
         \right)=\left(
                   \begin{array}{cccc}
                     x_{1,0} &  x_{1,1} & \cdots & x_{1,L-1} \\
                     x_{2,0} &  x_{2,1} & \cdots & x_{2,L-1} \\
                     \vdots & \vdots & \ddots & \vdots \\
                    x_{N_t,0} &  x_{N_t,1} & \cdots & x_{N_t,L-1} \\
                   \end{array}
                 \right)_{N_t\times L}.$$

Note that in GSM system, there are $N_{\text{active}}$ TAs activated over every time-slot.  Hence, $\Omega$ should be a sparse matrix where each training sequence $\mathbf{x}_n$ $(1\leq n\leq N_t)$ has $Q=L/\left\lceil N_t/N_{\text{active}}\right\rceil$ non-zero entries. In this paper, for simplicity, we suppose $N_{\text{active}}$ divides $N_t$. We consider the non-zero entries having identical magnitude of 1, and each training sequence has energy of $E=Q$.

Let $N_t=nN_{\text{active}}$ where $n$ is an integer with $0\leq n\leq D=N_t/N_{\text{active}}$. Let $\mathcal{S}=\{\mathcal{S}_1,\mathcal{S}_2,...,\mathcal{S}_{N_{\text{active}}}\}$ be a set of sequence sets, where $\mathcal{S}_k=\{\mathbf{s}_{k,1},\mathbf{s}_{k,2},...,\mathbf{s}_{k,D}\}$
and $\mathbf{s}_{k,m}=\left(s_{k,m}(0),s_{k,m}(1),...,s_{k,m}(Q-1)\right)$ for $1\leq k\leq N_{\text{active}}$ and $1\leq m\leq D$. For $1\leq m\leq D$, define\\
$$
 \mathcal{X}_m=\left(
     \begin{array}{c}
       \mathbf{s}_{1,m} \\
       \mathbf{s}_{2,m} \\
       \vdots \\
       \mathbf{s}_{N_{\text{active}},m} \\
     \end{array}
   \right)_{N_{\text{active}}\times Q}
$$ and
 $$\Omega=\left(
     \begin{array}{c}
       T^0\left(\mathcal{X}_1\mathbf{0}_{N_{\text{active}}\times\left(D-1\right)Q}\right) \\
   T^Q\left(\mathcal{X}_2\mathbf{0}_{N_{\text{active}}\times\left(D-1\right)Q}\right)\\
       \vdots \\
   T^{\left(D-1\right)Q}\left(\mathcal{X}_{D}\mathbf{0}
   _{N_{\text{active}}\times\left(D-1\right)Q}\right)\\
     \end{array}
   \right)_{N_t\times L}.
$$

An example of training matrix $\Omega$ having $N_t=4,~N_{\text{active}}=2$ and $L=2Q$ is shown below,
\begin{equation*}
 \Omega=\left(
     \begin{array}{cc}
       \mathbf{s}_{1,1}&\mathbf{0}\\
   \mathbf{s}_{2,1}&\mathbf{0}\\
   \mathbf{0}&\mathbf{s}_{1,2}\\
    \mathbf{0}&\mathbf{s}_{2,2}
     \end{array}
   \right)_{4\times 2Q},
\end{equation*}
where $\mathbf{0}$ denotes $\mathbf{0}_{1\times Q}$. Clearly, $\phi_{\mathbf{x}_1,\mathbf{x}_2}(u)$ and $\phi_{\mathbf{x}_3,\mathbf{x}_4}(u)$ are nonzero for $0\leq u\leq Q-1$, and $\phi_{\mathbf{x}_1,\mathbf{x}_3}(u),\phi_{\mathbf{x}_1,\mathbf{x}_4}(u),$ $\phi_{\mathbf{x}_2,\mathbf{x}_3}(u)$ and $\phi_{\mathbf{x}_2,\mathbf{x}_4}(u)$ are nonzero for $1\leq u\leq Q-1$. This implies that $\Omega$ cannot satisfy the optimal condition of GSM training sequences as Eq. (\ref{eq-MSE}).

To solve this problem, we consider $\mathcal{S}^j=\{\mathcal{S}^j_1,\mathcal{S}^j_2,$
$...,\mathcal{S}^j_{N_{\text{active}}}\}(1\leq j\leq J)$ be $J$ sets of sequence sets, where $\mathcal{S}^j_k=\{\mathbf{s}^j_{k,1},\mathbf{s}^j_{k,2},...,\mathbf{s}^j_{k,M}\}$
and $\mathbf{s}^j_{k,m}=\left(s^j_{k,m}(0),s^j_{k,m}(1),...,s^j_{k,m}(\theta-1)\right)$ is an unimodular sequence of length $\theta$ for $1\leq k\leq N_{\text{active}}$ and $1\leq m\leq M$, and the training matrix $\Omega$ with the following structure,
\begin{eqnarray}\label{eq-omi}
\Omega&=&(\Omega_1,\mathbf{0}_{N_t\times \lambda},\Omega_2,\mathbf{0}_{N_t\times \lambda},...,\Omega_J,\mathbf{0}_{N_t\times \lambda}),
\end{eqnarray}
\begin{eqnarray*}
\Omega_j&=& \left(
     \begin{array}{c}
       T^0\left(\mathcal{X}^j_1\mathbf{0}_{N_{\text{active}}\times\left(D-1\right)\theta}\right) \\
   T^\theta\left(\mathcal{X}^j_2\mathbf{0}_{N_{\text{active}}\times\left(D-1\right)\theta}\right)\\
       \vdots \\
   T^{\left(D-1\right)\theta}\left(\mathcal{X}^j_{D}\mathbf{0}
   _{N_{\text{active}}\times\left(D-1\right)\theta}\right)\\
     \end{array}
   \right)_{N_t\times D\theta},
\end{eqnarray*}
where for $1\leq n\leq D$,
$$
 \mathcal{X}_n^j=\left(
     \begin{array}{c}
       \mathbf{s}^j_{1,n} \\
       \mathbf{s}^j_{2,n} \\
       \vdots \\
       \mathbf{s}^j_{N_{\text{active}},n} \\
     \end{array}
   \right)_{N_{\text{active}}\times \theta}
$$
with $J\geq2$ is a positive integer and $\mathbf{0}_{N_t\times\lambda}$ in Eq. (\ref{eq-omi}) is called the $\lambda$-length zero-time slot of $\Omega$. Note that $DJ\theta+\lambda J=L$ and $J\theta=Q$.  In the sequel, we sometimes write the training matrix $\Omega$  as $(N_t,N_{\text{active}},\lambda,J,\theta)-\Omega$.

\begin{Example}
  Training matrix $(4,2,\lambda,2,\theta)-\Omega$ for channel estimation in GSM is shown in Fig. \ref{fig-eg-GSM} with
  \begin{align*}
    \Omega=&\left(
               \begin{array}{cccccc}
                 \mathbf{s}_{1,1}^1 & \mathbf{0}_{1\times\theta} &\mathbf{0}_{1\times\lambda} & \mathbf{s}_{1,1}^2 & \mathbf{0}_{1\times\theta} & \mathbf{0}_{1\times\lambda} \\
                 \mathbf{s}_{2,1}^1 & \mathbf{0}_{1\times\theta} & \mathbf{0}_{1\times\lambda} & \mathbf{s}_{2,1}^2 & \mathbf{0}_{1\times\theta} & \mathbf{0}_{1\times\lambda}\\
                 \mathbf{0}_{1\times\theta} & \mathbf{s}_{1,2}^1 & \mathbf{0}_{1\times\lambda} & \mathbf{0}_{1\times\theta} &\mathbf{s}_{1,2}^2 & \mathbf{0}_{1\times\lambda} \\
                 \mathbf{0}_{1\times\theta} & \mathbf{s}_{2,2}^1 & \mathbf{0}_{1\times\lambda} & \mathbf{0}_{1\times\theta} & \mathbf{s}_{2,2}^2 & \mathbf{0}_{1\times\lambda} \\
               \end{array}
             \right)\\
     =&  \left(
               \begin{array}{cccccc}
                 \mathcal{X}_1^1 & \mathbf{0}_{2\times\theta} &\mathbf{0}_{2\times\lambda} & \mathcal{X}_1^2 & \mathbf{0}_{2\times\theta} & \mathbf{0}_{2\times\lambda} \\
                 \mathbf{0}_{2\times\theta} & \mathcal{X}_2^1 & \mathbf{0}_{2\times\lambda} & \mathbf{0}_{2\times\theta} &\mathcal{X}_2^2 & \mathbf{0}_{2\times\lambda} \\
               \end{array}
             \right).
  \end{align*}
\begin{figure}[htbp]
  \centering
  \includegraphics[width=3in]{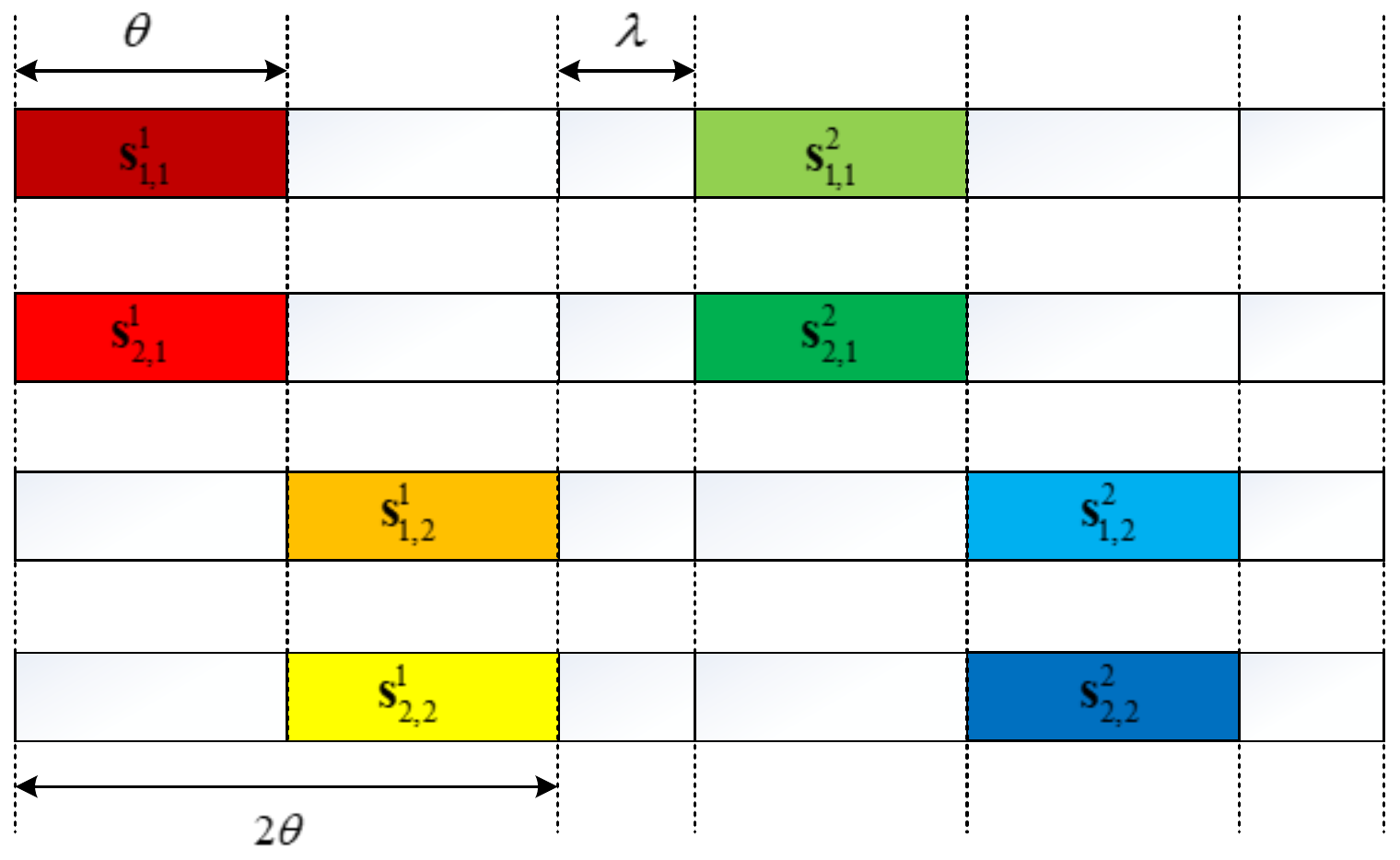}\\
  \caption{Training matrix with $N_t=4,N_{\text{active}}=2,J=2,L=4\theta+2\lambda$}\label{fig-eg-GSM}
\end{figure}

\end{Example}

According to Eq. (\ref{eq-MSE}), the optimal GSM training criteria under the proposed structure of training matrix in Eq. (\ref{eq-omi}) can be concluded in the following lemma.
\begin{Lemma}\label{lem-criteria}
  $\Omega$ is called an optimal GSM training sequence matrix in Eq. (\ref{eq-omi}) if and only if it satisfies the following three equations:
  \begin{align}
\nonumber
& \sum\limits_{j=1}^J\rho_{\mathbf{s}^j_{k,n}}(u)=0, \\\label{eq-condition-1}
 &\hbox{for}~1\leq u\leq\lambda, 1\leq k\leq N_{\text{active}}, 1\leq n\leq N_t/N_{\text{active}};\\
 \nonumber
&\sum\limits_{j=1}^J\rho_{\mathbf{s}^j_{k,n},\mathbf{s}^j_{k',n}}(u)=0,~\hbox{for}~
0\leq\mid u\mid\leq\lambda, 1\leq k\ne k'\leq N_{\text{active}},\\
\label{eq-condition-2}
 &\hspace{4cm}\hbox{and}~1\leq n\leq  N_t/N_{\text{active}};
  \end{align}
   \begin{align}
 \nonumber
&\sum\limits_{j=1}^J\rho_{\mathbf{s}^j_{k,n+1},\mathbf{s}^j_{k',n}}(\theta-u)=0,~\hbox{for}~ 1\leq u\leq\lambda,\\
\label{eq-condition-3}
 &\hspace{0.5cm}\hbox{and}~ 1\leq k\ne k'\leq N_{\text{active}},
~ 1\leq n\leq  N_t/N_{\text{active}}-1.
\end{align}
\end{Lemma}

It can be easily verified that the following remark proposed in \cite{Liu-Yang-20} is also valid for our proposed structure.
\begin{Remark}(Remark 5 of \cite{Liu-Yang-20})
   For any $q$-ary characteristic matrix with even $q$, $J$ should be even. This implies that any binary ($q = 2$) or quaternary ($q = 4$) characteristic matrix should have even $J$.
\end{Remark}

\textit{\textbf{Proposed Theorem}}: Consider a $(N_t,N_{\text{active}},\lambda,J,\theta)-\Omega$ training matrix as follows.
\begin{align}\label{eq-seed}
 \Omega=&\left(\Omega_1,\mathbf{0}_{N_t\times \lambda},\Omega_2,\mathbf{0}_{N_t\times \lambda},...,\Omega_J,\mathbf{0}_{N_t\times \lambda}\right),
      \end{align}
      where
      $$
\Omega_j=\left(
     \begin{array}{c}
       T^0\left(\mathcal{X}_j\mathbf{0}_{N_{\text{active}}\times\left(D-1\right)\theta}\right) \\
   T^\theta\left(\mathcal{X}_j\mathbf{0}_{N_{\text{active}}\times\left(D-1\right)\theta}\right)\\
       \vdots \\
   T^{\left(D-1\right)\theta}\left(\mathcal{X}_j\mathbf{0}
   _{N_{\text{active}}\times\left(D-1\right)\theta}\right)\\
     \end{array}
   \right)_{N_t\times D\theta}$$ and
   $$\mathcal{X}_j=\left(
        \begin{array}{c}
          \mathbf{a}_{1,j}\\
          \mathbf{a}_{2,j} \\
         \vdots \\
         \mathbf{a}_{N_{\text{active}},j}\\
        \end{array}
      \right)_{N_{\text{active}}\times\theta}
~~(1\leq j\leq J).$$

\begin{Theorem}\label{thm-opt-training}
The training matrix in Eq. (\ref{eq-seed}) is optimal if and only if $\{\mathcal{A}_1,\mathcal{A}_2,...,$ $\mathcal{A}_{N_{\text{active}}}\}$ is a $(N_{\text{active}},J,\theta,\lambda)$-SZCCS, where $\mathcal{A}_k=\{\mathbf{a}_{k,1},\mathbf{a}_{k,2},...,\mathbf{a}_{k,J}\}$ $(1\leq k\leq N_{\text{active}})$.
\end{Theorem}
\begin{proof}
Note that
$
\mathbf{s}^j_{k,n}=\mathbf{a}_{k,j},
 ~\hbox{for}~1\leq k\leq N_{\text{active}}, 1\leq n\leq N_t/N_{\text{active}},~\hbox{and}~1\leq j\leq J,
$
and so
\begin{align}\nonumber
&\rho_{\mathbf{s}^j_{k,n},\mathbf{s}^j_{k',n'}}(u)=\rho_{\mathbf{a}_{k,j},\mathbf{a}_{k',j}}(u),~\hbox{for}~ \forall u,
~ 1\leq k, k'\leq N_{\text{active}},\\\label{eq-thm}
&\hspace{2.5cm}\hbox{and}
~ 1\leq n,n'\leq  N_t/N_{\text{active}}.
\end{align}
Hence, Eqs. (\ref{eq-condition-1})-(\ref{eq-condition-3}) in \textbf{Lemma \ref{lem-criteria}} can be further expressed as follows.
\begin{eqnarray}
\nonumber
 \sum_{j=1}^J\rho_{\mathbf{a}_{k,j}}(u)=0,\hspace{2.5cm}\\\label{eq-result-c1}
 \hbox{for~}u\in\{1,...,\lambda\}\bigcup\{\theta-
 \lambda,...,\theta-1\},1\leq k\leq N_{\text{active}};\\
\nonumber
 \sum_{j=1}^J\rho_{\mathbf{a}_{k,j},\mathbf{a}_{k',j}}(u)=0,\hspace{2.2cm}\\
 \nonumber
\hbox{for~}u\in\{0\}\bigcup
 \{1,...,\lambda\}\bigcup\{\theta-
 \lambda,...,\theta-1\},\\ \label{eq-result-c2}1\leq k\ne k'\leq N_{\text{active}}.
\end{eqnarray}
Combining Eqs. (\ref{eq-result-c1})-(\ref{eq-result-c2}) and \textbf{Definition \ref{def-SZCCS}}, it can be observed that $\{\mathcal{A}_1,\mathcal{A}_2,...,$ $\mathcal{A}_{N_{\text{active}}}\}$ is a $(N_{\text{active}},J,\theta,\lambda)$-SZCCS.
\end{proof}

The training matrix $\Omega$ can be enlarged when $J'$ is a multiple of $J$, i.e., $J'=tJ$ for some integer $t$, $\Omega$ can be enlarged to $(N_t,N_{\text{active}},J',\theta)-\Omega'$ as
$$\Omega'=\left(\underbrace{\begin{array}{cccc}
         \Omega & \Omega & \cdots & \Omega\\
        \end{array}}_{t\times\Psi}
        \right)_{N_t\times tJ(D\theta+\lambda)}.
$$

\subsection{Minimum MSE for the Proposed Training Framework}

In this subsection, we show the minimum MSE related to the sequence length of $\mathcal{S}$, $E_b/N_0$ and the numbers of transmit antennas and activated antennas and multi-paths which are denoted by $N_t$, $N_{\text{active}}$ and $\lambda+1$, respectively, under the training framework proposed in \emph{Subsection A} with each training sequence having energy $E$.

Note that the sequence length of training sequence is $L= N_t/N_{\text{active}}\times(\lambda+\theta)$. Then, the energy per bit is
$E_b=\frac{E}{L}=\frac{E}{ N_t/N_{\text{active}}\times(\lambda+\theta)},$
where
\begin{equation}\label{MSE-EBN0}
\frac{E_b}{N_0}=\frac{E}{2\sigma^2_{\omega} N_t/N_{\text{active}}\times(\lambda+\theta)},
\end{equation}
$N_0=2\sigma^2_{\omega}$ is the power spectral density of the white complex Gaussian noise. Combining Eqs. (\ref{MSE}) and (\ref{MSE-EBN0}), we have that
\begin{equation}\label{MSE-new}
 \hbox{minimum~MSE}=\left(2 \left(N_t/N_{\text{active}}\right)(\lambda+\theta)E_b /N_0\right)^{-1}.
\end{equation}
Therefore, for given $E_b /N_0$ and $N_t$, the minimum MSE decreases with $\lambda$ and $\theta$, and increases with $N_{\text{active}}$.
\subsection{Numerical Evaluation}

In this subsection, we analyse the performance of the proposed SZCCSs in \textbf{Therorem \ref{thm-optimalcon}} and \textbf{Theorem \ref{thm-SZCCS-non1}} as the training sequences for GSM systems over frequency-selective channels, based on the framework proposed in the \emph{Subsection \ref{sub-framework}} for different settings of $E_b/N_0$, number of multi-paths $\lambda+1$, number of activated antennas $N_{\text{active}}$ and number of zero time-slots $\theta$. Consider the same frequency-selective channel in \cite{Liu-Yang-20}, where the $\lambda+1$-path channel (separated by integer symbol durations) has uniform power delay profile as
$
h[t]=\sum_{n=0}^{\lambda} h_{i} \delta[t-n T]
$
where $h_{i}$'s are complex-valued Gaussian random variables with
zero mean and $\mathbb{E}\left(\left|h_{i}\right|^{2}\right)=1$.

\subsubsection{MSE comparison with the change of $E_b/N_0$}
We consider a generic training-based single carrier GSM transmission structure with $N_t=4$ TAs, one receive antenna and $N_{\text{active}}=2$ RF chains, and a 6-path channel. Let the training matrix of the random binary sequences for $J=2$ be
$\label{eq-struc}
    \Omega_r
             =\left(
               \begin{array}{cccccc}
                 \mathbf{s}_{r,1} & \mathbf{0}_{1\times32} &\mathbf{0}_{1\times5} & \mathbf{s}_{r,2} & \mathbf{0}_{1\times32} & \mathbf{0}_{1\times5} \\
                 \mathbf{s}_{r,3} & \mathbf{0}_{1\times32} & \mathbf{0}_{1\times5} & \mathbf{s}_{r,4} & \mathbf{0}_{1\times32} & \mathbf{0}_{1\times5}\\
                 \mathbf{0}_{1\times32} & \mathbf{s}_{r,1}& \mathbf{0}_{1\times5} & \mathbf{0}_{1\times32} &\mathbf{s}_{r,2} & \mathbf{0}_{1\times5} \\
                 \mathbf{0}_{1\times32} & \mathbf{s}_{r,3} & \mathbf{0}_{1\times5} & \mathbf{0}_{1\times32} & \mathbf{s}_{r,4} & \mathbf{0}_{1\times5} \\
               \end{array}
             \right),
 $
     where $\mathbf{s}_{r,1},\mathbf{s}_{r,2},\mathbf{s}_{r,3},\mathbf{s}_{r,4}$ are randomly generated binary sequences of length 32.  The aperiodic correlation sums of $\{\mathbf{s}_{r,1},\mathbf{s}_{r,2},\mathbf{s}_{r,3},\mathbf{s}_{r,4}\}$ is presented in Fig. \ref{fig-ran}. It can be observed that there is no ZCZ in the these correlation sums, which implies $\Omega_r$ can not satisfy the optimal conditions in Eqs. (\ref{eq-condition-1})-(\ref{eq-condition-3}).
     Let $\Omega_{s}$ be the training matrix of the proposed SZCCSs for $J=2$, which can be generated by replacing $\mathbf{s}_{r,1},\mathbf{s}_{r,2},\mathbf{s}_{r,3},\mathbf{s}_{r,4}$ by $\mathbf{s}_{s,1},\mathbf{s}_{s,2},\mathbf{s}_{s,3},\mathbf{s}_{s,4}$, respectively,
     where $\mathbf{s}_{s,1},\mathbf{s}_{s,2},\mathbf{s}_{s,3},\mathbf{s}_{s,4}$ are the sequences in the first two sequence sets of the proposed optimal $(8,2,32,7)$-SZCCS generated in \textbf{Example \ref{eg-optimal}}, where
      \begin{eqnarray*}
       \mathbf{s}_{s,1} &=& (+,	+,	+,	-,	+,	+,	-,	+,	+,	+,	+,	-,	-,	-,	+,	-,	+,\\&&	+,	+,	-,	+,	+,	-,	+,	+,	+,	+,	-,	-,	-,	+,	-),\\
       \mathbf{s}_{s,2} &=& (+,	-,	+,	+,	+,	-,	-,	-,	+,	-,+,	+,	-,	+,	+,	+,	+,\\&&	-,	+,	+,	+,	-,	-,	-,	+,	-,	+	,+	,-,+,	+,	+),\\
       \mathbf{s}_{s,3} &=& (+,	+,	+,	-,	-,	-,	+,	-,	-,	-,	-,	+,	-,	-,	+,	-,	+,\\&&	+,	+,	-,	-,	-,	+,	-,	-,	-,	-,	+,	-,		-,+,	-), \\
       \mathbf{s}_{s,4} &=& (+,	-,	+,	+,	-,	+,	+,	+,	-,	+,	-,	-,	-,	+,	+,	+,	+,\\&&	-,	+,	+,	-,	+,	+,	+,	-,	+,	-,	-,	-,	+,	+,	+).
     \end{eqnarray*}The aperiodic correlation sums of $\{\mathbf{s}_{s,1},\mathbf{s}_{s,2},\mathbf{s}_{s,3},\mathbf{s}_{s,4}\}$ is presented in Fig. \ref{fig-szccs}, where shows that there are symmetric ZCZs of width greater or equal to 7 in these aperiodic correlation sums. Hence,

     $\{\mathbf{s}_{s,1},\mathbf{s}_{s,2},\mathbf{s}_{s,3},\mathbf{s}_{s,4}\}$ satisfies the optimal conditions in Eqs. (\ref{eq-condition-1})-(\ref{eq-condition-3}) when $\lambda\leq7$. The argument for the cases that $J=6,18$ are the same. In Fig. \ref{fig-SNR} we evaluate the channel estimation MSE performances of $\Omega_{r}$ and $\Omega_{s}$ when $J=2,6,18$ and the $E_b/N_0$ runs over $\{0,2,...,20\}$.
     It can be observed that  when the number of multi-paths is 6, the proposed $(8,2,32,7)$-SZCCS can be used to design optimal GSM training matrix which attains the MSE lower bound, which outperforms that of random binary sequences. Those performances are consistent with the discussions of Figs. \ref{fig-ran} and \ref{fig-szccs} above. Moreover, the MSE decreases when $J$ increases, which can be seen as the length of training sequences increases.
\begin{figure}[!htbp]
  \centering
\begin{minipage}[c]{0.45\textwidth}
  \centering
  \includegraphics[width=3.25in]{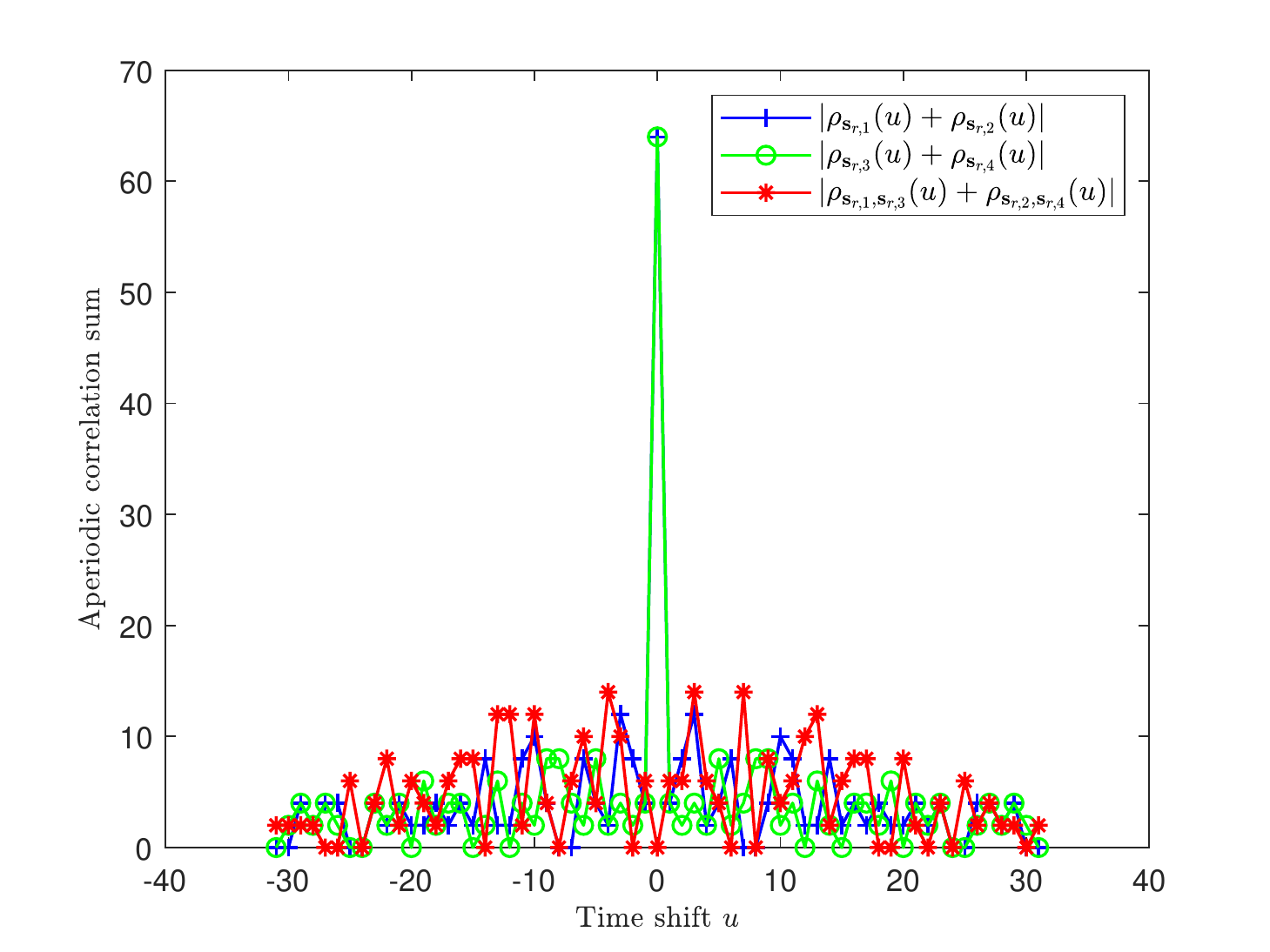}\\
  \caption{Aperiodic correlation sums of random training matrix}\label{fig-ran}
\end{minipage}
\begin{minipage}[c]{0.45\textwidth}
  \centering
  \includegraphics[width=3.25in]{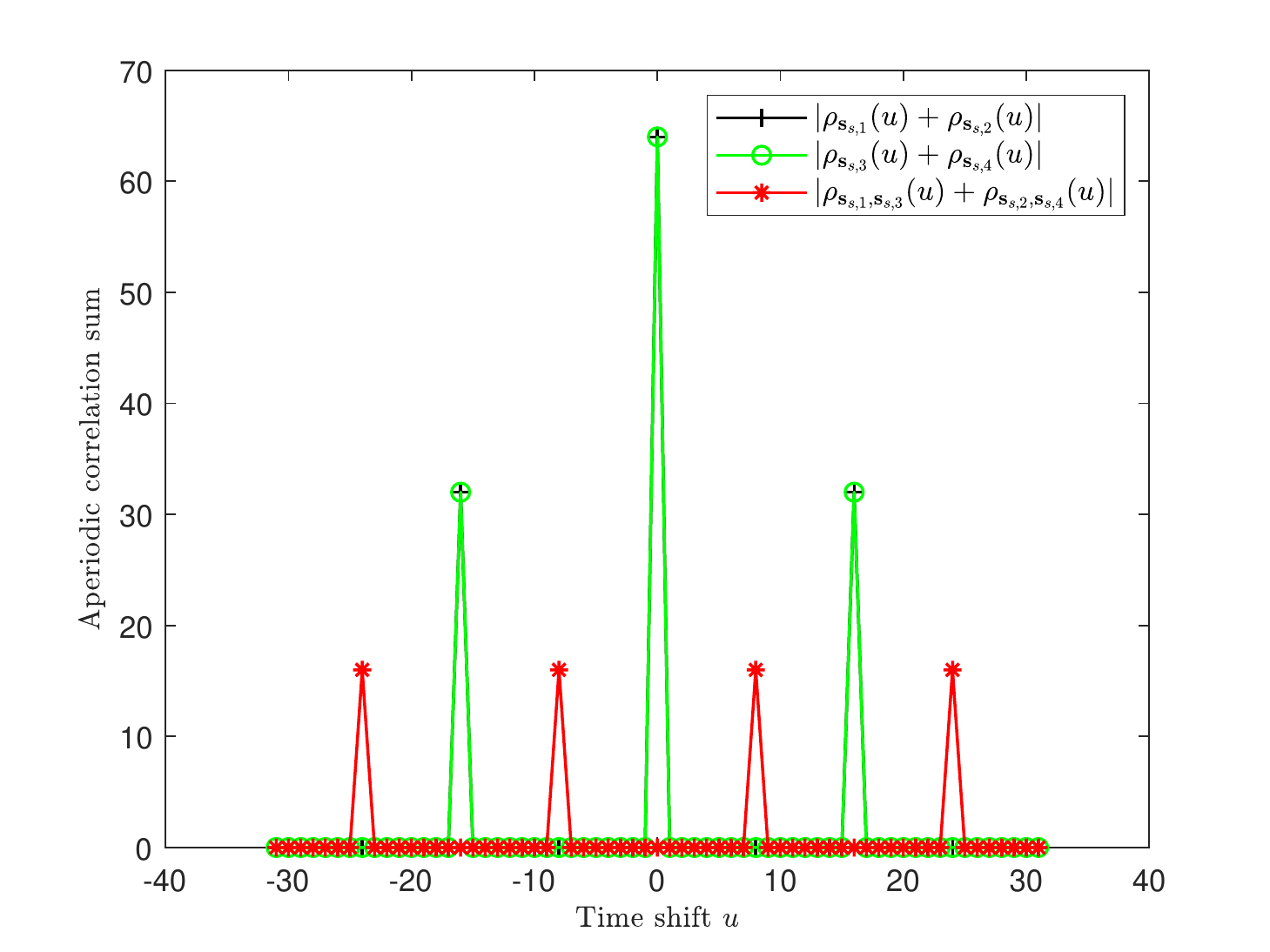}\\
  \caption{Aperiodic correlation sums of the training matrix of the proposed SZCCS}\label{fig-szccs}
  \end{minipage}
\end{figure}
\begin{figure}[htbp]
  \centering
  \begin{minipage}[c]{0.45\textwidth}
  \centering
  \includegraphics[width=3.25in]{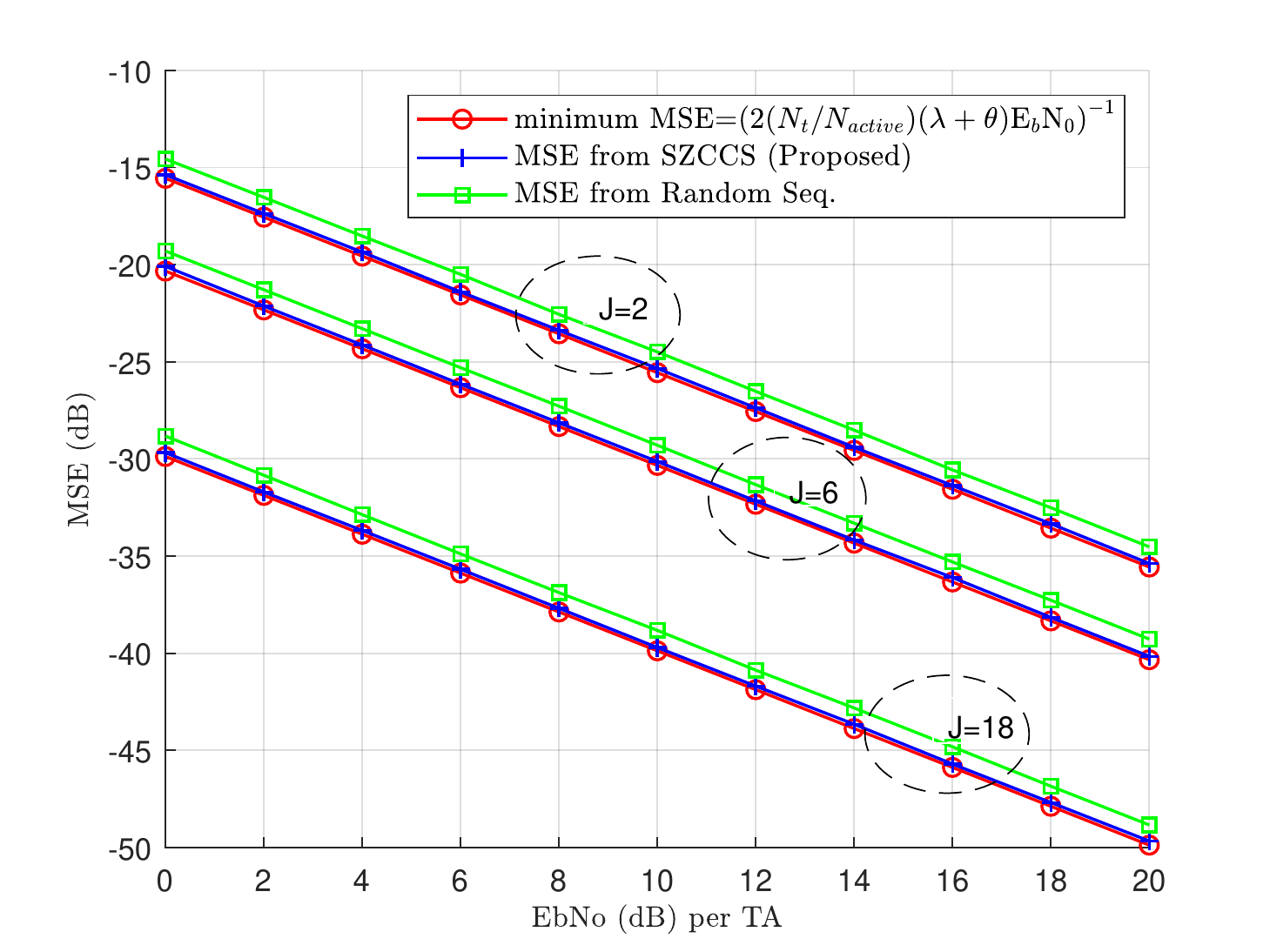}\\
  \caption{MSE comparison, for $(8,2,16,7)$-SZCCS , where the number of multi-paths is 6, $N_t=4,N_{\text{active}}=2$ and $J\in\{2,6,18\}$}\label{fig-SNR} \end{minipage}
\begin{minipage}[c]{0.45\textwidth}
  \centering
  \includegraphics[width=3.25in]{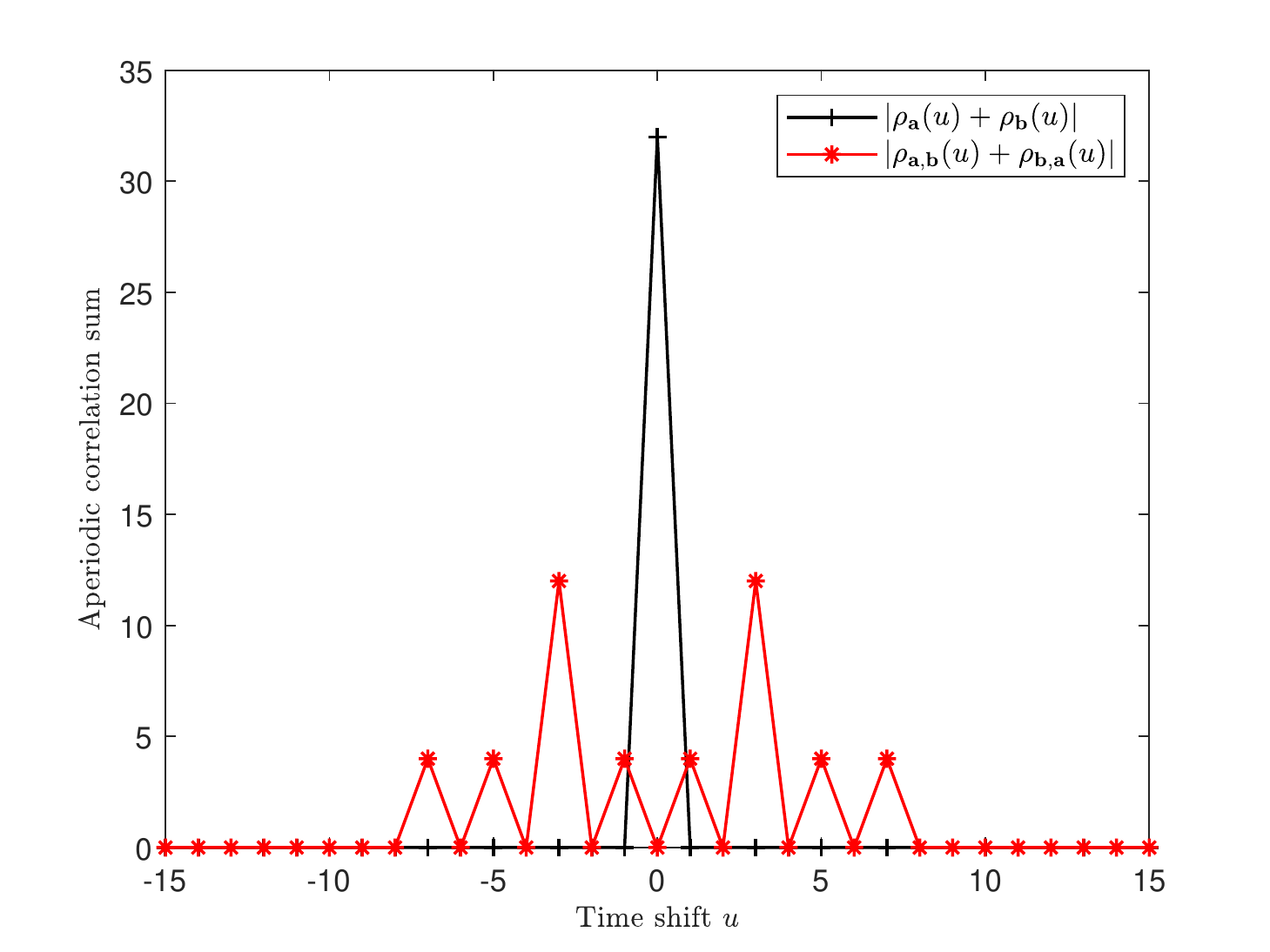}\\
  \caption{Aperiodic correlation sums of the training matrix of (16,8)-CZCP}\label{fig-czcp}
    \end{minipage}
\end{figure}

\subsubsection{MSE comparison with the change of the number of multi-paths}
We evaluate the channel estimation MSE performances for different values of multi-paths for $E_b/N_0=16$ dB with 4 transmit antennas and 2 RF chains in Fig. \ref{fig-multipath}. We employ the first two sequence sets of the optimal (8,2,32,7)-SZCCS constructed through \textbf{Theorem \ref{thm-optimalcon}}  and the (2,2,24,7)-SZCCS in \textbf{Theorem \ref{thm-SZCCS-non1}}  to generate our GSM training matrices. We compare  their channel estimation performances with those of the (16,8)-CZCP given in \cite{Liu-Yang-20}, the Zadoff-Chu sequences of length 32 and binary random sequences of length 32.
Let the training matrix of (16,8)-CZCP $\{\mathbf{a},\mathbf{b}\}$ be
$$
    \Omega_c
             =\left(
               \begin{array}{cccccc}
                 \mathbf{a} & \mathbf{0}_{1\times16} &\mathbf{0}_{1\times\lambda} & \mathbf{b} & \mathbf{0}_{1\times16} & \mathbf{0}_{1\times\lambda} \\
                 \mathbf{b} & \mathbf{0}_{1\times16} & \mathbf{0}_{1\times\lambda} & \mathbf{a} & \mathbf{0}_{1\times16} & \mathbf{0}_{1\times\lambda}\\
                 \mathbf{0}_{1\times16} & \mathbf{a}& \mathbf{0}_{1\times\lambda} & \mathbf{0}_{1\times16} &\mathbf{b} & \mathbf{0}_{1\times\lambda} \\
                 \mathbf{0}_{1\times16} & \mathbf{b} & \mathbf{0}_{1\times\lambda} & \mathbf{0}_{1\times16} & \mathbf{a}& \mathbf{0}_{1\times\lambda} \\
               \end{array}
             \right),
 $$
     where
     \begin{eqnarray*}
      \mathbf{a}&=&(1, 1, 1, -1, 1, 1, -1 ,1 , 1, -1, 1, 1, 1,-1, -1, -1)\\
     \mathbf{b}&=&(1, 1, 1, -1, 1, 1, -1, 1,-1, 1,-1, -1 ,-1,1, 1, 1).
     \end{eqnarray*}
 The aperiodic correlation sums of $\{\mathbf{a},\mathbf{b}\}$ is shown in Fig. \ref{fig-czcp}, it can be observed that $\{\mathbf{a},\mathbf{b}\}$ satisfies optimal conditions in Eqs. (\ref{eq-condition-1}) and (\ref{eq-condition-3}) without satisfying Eq. (\ref{eq-condition-2}) when $\lambda\leq8$. Let the training matrix of four Zadoff-Chu sequences $\mathbf{s}_{zc,1},\mathbf{s}_{zc,2},\mathbf{s}_{zc,3},\mathbf{s}_{zc,4}$ be
$$
    \Omega_{zc}
             =\left(
               \begin{array}{cccccc}
                 \mathbf{s}_{zc,1} & \mathbf{0}_{1\times32} &\mathbf{0}_{1\times\lambda} & \mathbf{s}_{zc,2} & \mathbf{0}_{1\times32} & \mathbf{0}_{1\times\lambda} \\
                 \mathbf{s}_{zc,3} & \mathbf{0}_{1\times32} & \mathbf{0}_{1\times\lambda} &\mathbf{s}_{zc,4} & \mathbf{0}_{1\times32} & \mathbf{0}_{1\times\lambda}\\
                 \mathbf{0}_{1\times32} & \mathbf{s}_{zc,1}& \mathbf{0}_{1\times\lambda} & \mathbf{0}_{1\times32} &\mathbf{s}_{zc,2} & \mathbf{0}_{1\times\lambda} \\
                 \mathbf{0}_{1\times32} & \mathbf{s}_{zc,3} & \mathbf{0}_{1\times\lambda} & \mathbf{0}_{1\times32} & \mathbf{s}_{zc,4}& \mathbf{0}_{1\times\lambda} \\
               \end{array}
             \right),
     $$
     where
      $\mathbf{s}_{zc,1},\mathbf{s}_{zc,2},\mathbf{s}_{zc,3},\mathbf{s}_{zc,4}$ are Zadoff-Chu sequences of length 32 generated with roots being 1, 3, 5, 7, respectively.
 The aperiodic correlation sums of $\{\mathbf{s}_{zc,1},\mathbf{s}_{zc,2},\mathbf{s}_{zc,3},\mathbf{s}_{zc,4}\}$ is shown in Fig. \ref{fig-czcp}, which implies that $\{\mathbf{s}_{zc,1},\mathbf{s}_{zc,2},\mathbf{s}_{zc,3},$ $\mathbf{s}_{zc,4}\}$ has no ZCZs and can not satisfy optimal conditions in Eqs. (\ref{eq-condition-1})-(\ref{eq-condition-3}) for any $\lambda>0$.
 Fig. \ref{fig-multipath} shows that the proposed SZCCSs lead to optimal GSM training matrices which attain the MSE lower bound when the  number of multi-paths $\lambda+1$ is less than or equal to $Z+1$. Note that the minimum MSE decreases with the number of multi-paths $\lambda+1$, which is consistent with Eq. (\ref{MSE-new}). It is noted that the MSE performance corresponding to CZCP, which is proposed in \cite{Liu-Yang-20} for optimal training of SM systems, is close to that of the Zadoff-Chu sequence and binary random sequence, indicating that CZCP may not be suitable for training matrix design in GSM systems.
\begin{figure}[!htbp]
  \centering
  \begin{minipage}[c]{0.45\textwidth}
  \centering
  \includegraphics[width=3.25in]{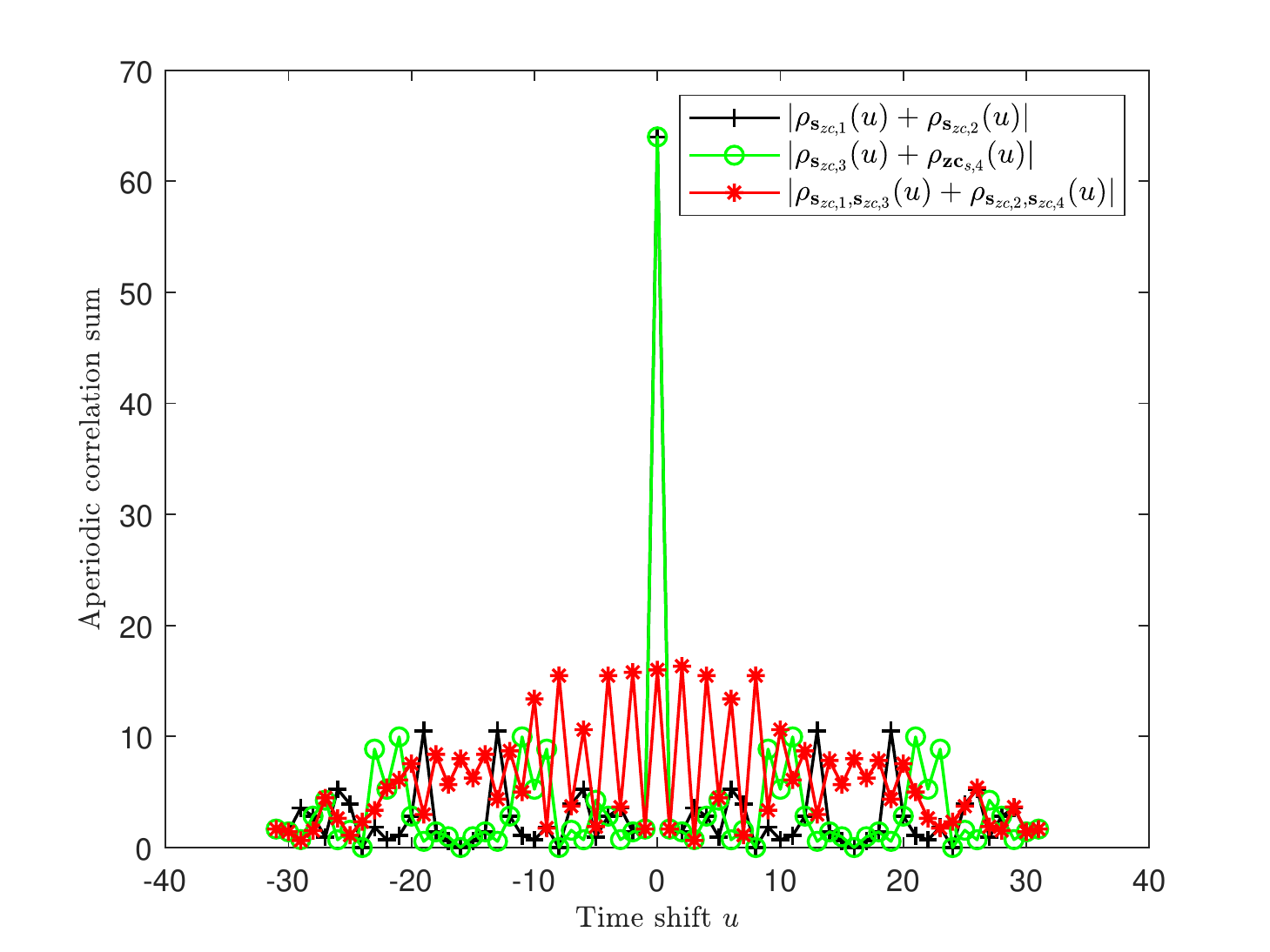}\\
  \caption{Aperiodic correlation sums of the training matrix of Zadoff-Chu sequences}\label{fig-zc}
\end{minipage}
  \begin{minipage}[c]{0.45\textwidth}
  \centering
  \includegraphics[width=3.25in]{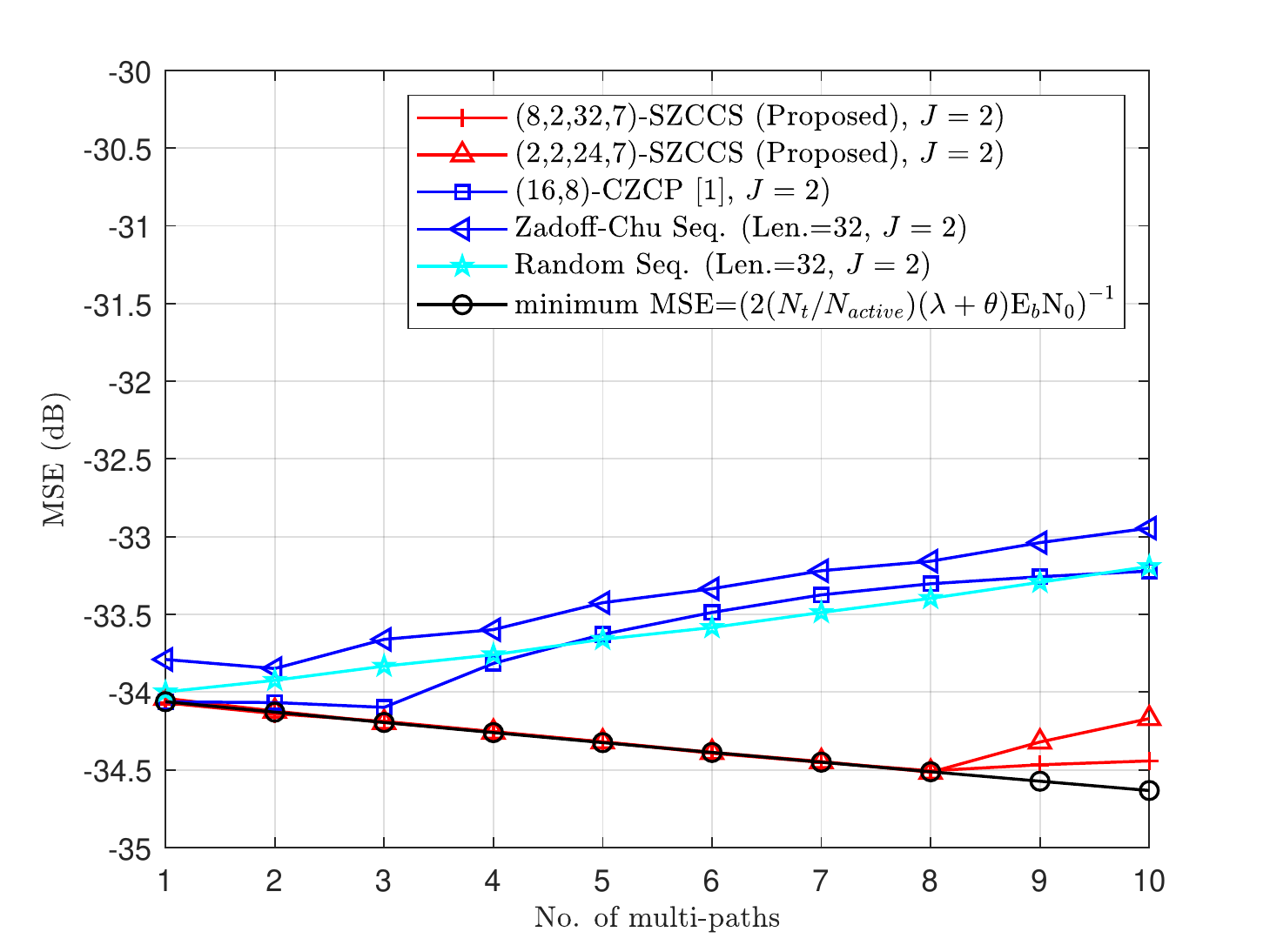}\\
  \caption{MSE comparison among different sequences for various multi-paths with $N_t=4,N_{\text{active}}=2$ and $\text{E}_b\text{N}_0$=16 dB.}\label{fig-multipath}
\end{minipage}
\end{figure}
\begin{figure}[!htbp]
  \centering
  \begin{minipage}[c]{0.45\textwidth}
  \centering
  \includegraphics[width=3.25in]{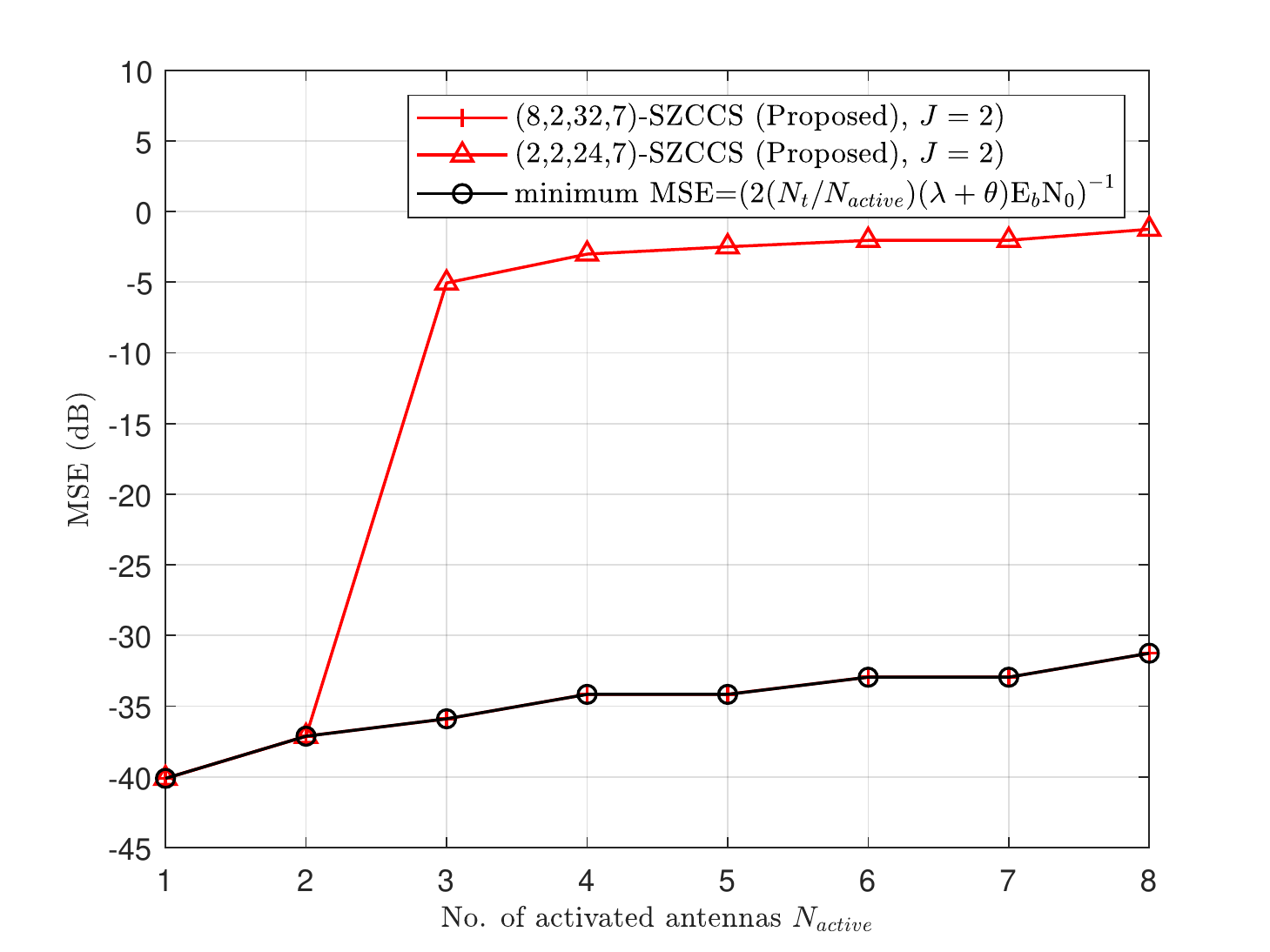}\\
  \caption{MSE comparison of the proposed proposed SZCCSs for various $N_{\text{active}}$ with $E_b/N_0$=16 dB and $\lambda=3$.}\label{fig-actived}
\end{minipage}
  \begin{minipage}[c]{0.45\textwidth}
  \centering
  \includegraphics[width=3.25in]{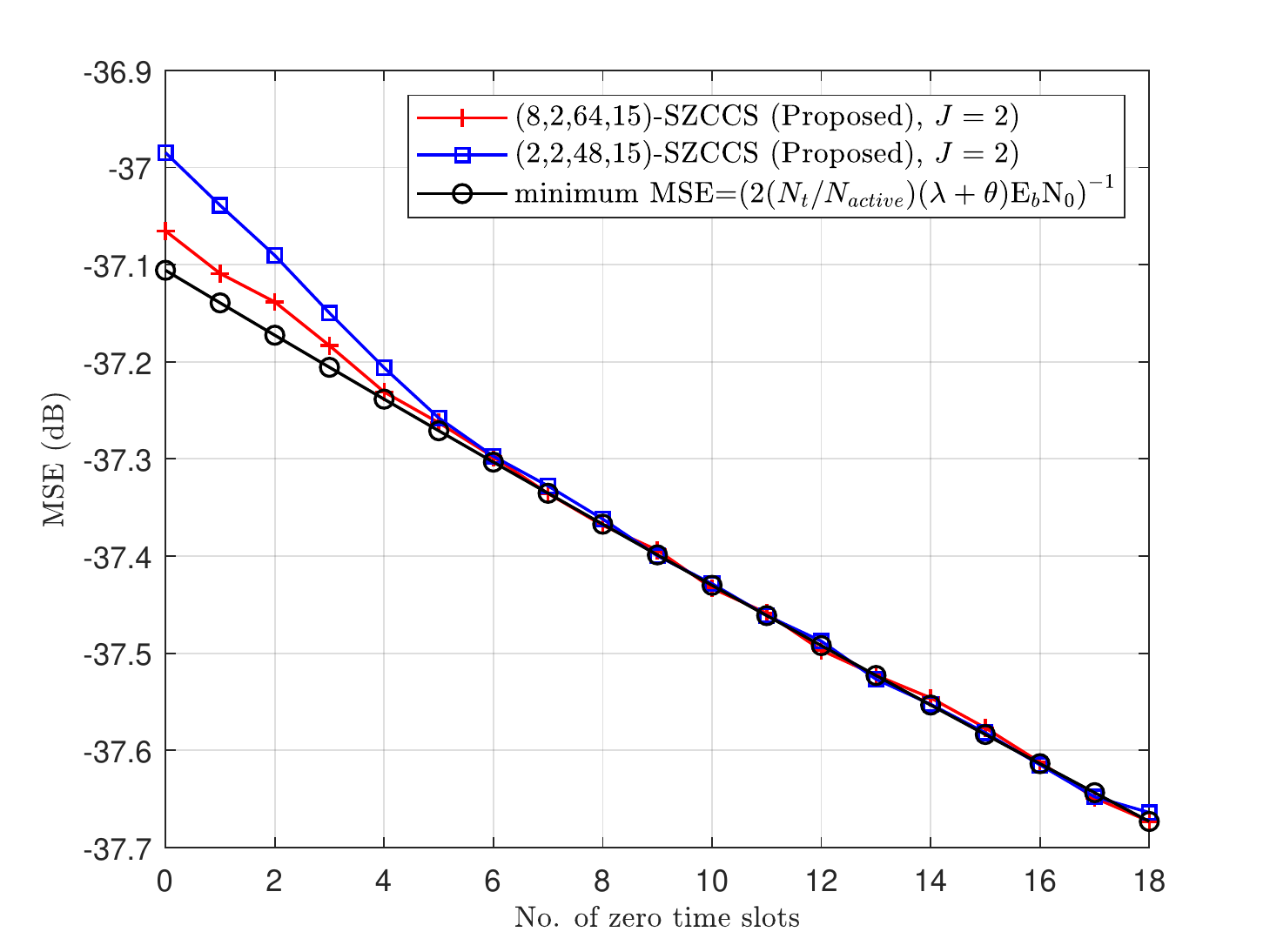}\\
  \caption{MSE comparison among the proposed SZCCSs for different values of zero-time slot with $E_b/N_0$=16 dB and $\lambda=15$.}\label{fig-zero}
  \end{minipage}
\end{figure}
\subsubsection{MSE comparison with the change of the number of activated antennas}
 Fig. \ref{fig-actived} shows that every GSM training matrix corresponding to  the proposed SZCCSs attains the MSE lower bound when the number of the activated antennas $N_{\text{active}}$ is less than or equal to the size $K$ of the SZCCS, and the minimum MSE increases with the number of activated antennas $N_{\text{active}}$.

\subsubsection{MSE comparison with the change of the number of zero time-slots}Finally, we evaluate the channel estimation MSE performances of the proposed $(8,2,64,15)$-SZCCS and $(2,2,48,15)$-SZCCS for different values of zero-time slots with fixed $\lambda=15$ in Fig. \ref{fig-zero}. One observes that reducing or removing the zero-time slot has little influence on the performance of the proposed GSM training matrices. Hence, in practical application, the zero-time slot may be removed to improve the spectral efficiency.
\section{Concluding Remarks}

In this paper, we have proposed a new class of code sets called symmetrical Z-complementary code sets which represents a stringent subclass of Z-complementary code sets. Unlike the conventional Z-complementary code sets where the aperiodic auto-correlation sum and aperiodic cross-correlation sum only have front-end zero-correlation zone properties, symmetrical Z- complementary code sets have an additional tail-end zero-correlation zone property. With the aid of generalized Boolean function, we have constructed two classes of symmetrical Z-complementary code sets including an optimal one with respect to the set size upper bound.

Moreover, we have shown that symmetrical Z-complementary code sets can be employed as training sequences for the optimal channel estimation in broadband generalized spatial modulation systems. We point out that the training sequences for conventional multiple-input multiple-output systems where the number of ratio-frequency chains equals that of transmit antennas is inapplicable in generalized spatial modulation. The same may be said for spatial modulation systems with a single ratio-frequency chain. We have introduced a generic design of optimal generalized spatial modulation training matrix using symmetrical Z-complementary code sets, where the inter-symbol interference and inter-channel interference can be minimized by the zero aperiodic auto-correlation sums and aperiodic cross-correlation sums in front-end and tail-end zero-correlation zones. Note that the number of ratio-frequency chains that can be supported is limited by the size of the applied symmetrical Z-complementary code sets, and the maximum tolerable number of multi-paths depends on the width of zero-correlation zones of aperiodic correlation sums for that symmetrical Z-complementary code set. We have shown that the designed training matrix leads to minimum channel estimation mean-squared-error in frequency-selective channels.
\appendices
\section{Proof of \textbf{Theorem \ref{thm-optimalcon}}}
\begin{proof}
  According to Eq. (\ref{eq-u}) and \textbf{Definition \ref{def-SZCCS}}, this proof can be divided into two steps. In the first step, it should be demonstrated that
  \begin{eqnarray}\nonumber
   C_{\mathcal{S}_k}(u)=\rho_{\psi(\mathbf{a}_k)}(u)+\rho_{\psi(\mathbf{b}_k)}(u)=0,\\
   \label{proof-au}\hbox{for}~u\in \mathcal{T}_1\bigcup\mathcal{T}_2,~1\leq k\leq8,
  \end{eqnarray}
  where $\mathcal{T}_1=\{1,2,\cdots,2^{m-2}-1\}$ and $\mathcal{T}_2=\{2^m+1-2^{m-2},2^m+2-2^{m-2},\cdots,2^m-1\}$. In the second step, it needs to prove that
   \begin{eqnarray}\nonumber
    C_{\mathcal{S}_k,\mathcal{S}_n}(u)=\rho_{\psi(\mathbf{a}_k),\psi(\mathbf{a}_n)}(u)
    +\rho_{\psi(\mathbf{b}_k),\psi(\mathbf{b}_n)}(u)=0,\\
   \label{eq-cro}\hbox{for}~u\in \{0\}\bigcup\mathcal{T}_1\bigcup\mathcal{T}_2,~1\leq k\neq n\leq8.
  \end{eqnarray}
   For any integers $i$ and $u$, let $j=i+u$; also let $(i_{1},i_{2},\cdots,i_{m})$ and $(j_{1},j_{2},\cdots,j_{m})$ be the binary representations of $i$ and $j$, respectively. Without loss of generality, we only discuss the cases when $k,n=1,2$ here.
  \begin{description}
    \item[Step 1]: When $\mathbf{D}_1=(0,0,0,0)$, it is clear that $a_1(\mathbf{x})=f(\mathbf{x})$, then, by \textbf{Lemma \ref{lem-GCP}}, we have that $\mathcal{S}_1=\{\psi(\mathbf{a}_1),\psi(\mathbf{b}_1)\}$ is a GCP of $2^m$ and
          \begin{align*}
   &\rho_{\psi(\mathbf{a}_1)}(u)+\rho_{\psi(\mathbf{b}_1)}(u)=0,\hspace{1cm}\\&\hbox{for}~|u|\in \{1,2,\cdots,2^m-1\}\supseteq\mathcal{T}_1\bigcup \mathcal{T}_2.
  \end{align*}
  When $\mathbf{D}_2=(1,0,1,0)$, then, we have
  \begin{align*}
  a_2(\mathbf{x})=& \frac{q}{2}\sum_{s=1}^{m-2}x_{\pi(s)}
  x_{\pi(s+1)}+\sum_{s=1}^{m}\mu_{s}
  x_{s}+\mu
\\
    &+\left(\frac{q}{2} x_{\pi(m-2)}+\frac{q}{2}x_{\pi(m-2)}x_{\pi(m)}\right)\\
  &+\left(\frac{q}{2} x_{\pi(m)}+\frac{q}{2}x_{\pi(m-1)}x_{\pi(m)}\right),\\
   b_2(\mathbf{x})=& a_2(\mathbf{x})
  +\frac{q}{2}x_{\pi(1)}.
\end{align*}
  For $u\in \mathcal{T}_1\bigcup \mathcal{T}_2$, Eq. (\ref{proof-au}) can be further expressed as
  \begin{eqnarray}\label{eq-AAF-1}
   \sum_{i=0}^{2^{m}-1-u}\left(\xi^{a_2(i)-a_2(j)}
   +
   \xi^{b_2(i)-b_2(j)}\right)=0,
  \end{eqnarray}
 where
  \begin{align*}
   a_2(i)-a_2(j) =& \frac{q}{2}\sum_{s=1}^{m-2}\left(i_{\pi(s)}
  i_{\pi(s+1)}-j_{\pi(s)}
  j_{\pi(s+1)}\right)\\&+\sum_{s=1}^{m}\mu_{s}
  \left(i_{s}-j_{s}\right)\\&+\frac{q}{2} (i_{\pi(m-2)}-j_{\pi(m-2)})\\
  &+\frac{q}{2}(i_{\pi(m-2)}i_{\pi(m)}-j_{\pi(m-2)}j_{\pi(m)})
  \end{align*}
    \begin{align*}
  &+\frac{q}{2} (i_{\pi(m)}-j_{\pi(m)}) \\
  &+\frac{q}{2}(i_{\pi(m-1)}i_{\pi(m)}-j_{\pi(m-1)}j_{\pi(m)}),\\
 b_2(i)-b_2(j) =&  a_2(i)-a_2(j)+\frac{q}{2}\left(i_{\pi(1)}-j_{\pi(1)}\right).
  \end{align*}
  If $i_{\pi(1)}\ne j_{\pi(1)}$, it is straightforward that $\xi^{a_2(i)-a_2(j)}=-\xi^{b_2(i)-b_2(j)}$ and $\xi^{a_2(i)-a_2(j)}+\xi^{b_2(i)
        -b_2(j)}=0.$

  If $i_{\pi(1)}=j_{\pi(1)}$, we have $a_2(i)-a_2(j) = b_2(i)-b_2(j)$. Let $t$ be the smallest integer such that $i_{\pi(t)}\ne j_{\pi(t)}$, obviously, $2\leq t\leq m-2$, otherwise, since $\{\pi(1),\pi(2),\cdots,\pi(m-2)\}=\{1,2,\cdots,m-2\}$, then, we have
  $u\in\{2^{m-2},2^{m-1},2^{m-1}+2^{m-2}\}$, which is contradict with $u\in\mathcal{T}_1\bigcup\mathcal{T}_2$.
        Let $i'$ and $j'$ be integers which are different from $i$ and $j$ in only one position $\pi(t-1)$, i.e., $i'_{\pi(t-1)}=1-i_{\pi(t-1)}$ and $j'_{\pi(t-1)}=1-j_{\pi(t-1)}$ respectively, and so $j'=i'+u$, $0\leq i',j'\leq 2^{m}-1$ and $i'_{\pi(1)}=j'_{\pi(1)}$.
Note that
\begin{align*}
a_{2}(i')-a_{2}(j')=&
a_{2}(i)-a_{2}(j)
+\frac{q}{2}
(1-2i_{\pi(t-1)})i_{\pi(t)}\\
  & -
\frac{q}{2}(1-2j_{\pi(t-1)})j_{\pi(t)}\\
\equiv &a_{2}(i)-a_{2}(j)+{q\over 2}~~(\bmod~q).
\end{align*}
This implies
$\xi^{a_2(i)-a_2(j)}+
\xi^{a_2(i')-a_2(j')} =\xi^{b_2(i)-b_2(j)}+
\xi^{b_2(i')-b_2(j')} = 0.
$
Combining these cases above, we have Eq. (\ref{eq-AAF-1}) holds.
    \item[Step 2]: Eq. (\ref{eq-cro}) can be further expressed as
     \begin{eqnarray}\label{eq-ACF-12}
   \sum_{i=0}^{2^{m}-1-u}\left(\xi^{a_1(i)-a_2(j)}
   +
   \xi^{b_1(i)-b_2(j)}\right)=0,
  \end{eqnarray}
  where
  \begin{align*}
 a_1(i)-a_2(j)= &\frac{q}{2}\sum_{s=1}^{m-2}\left(i_{\pi(s)}
  i_{\pi(s+1)}-j_{\pi(s)}
  j_{\pi(s+1)}\right) \\
  & -\frac{q}{2}j_{\pi(m-2)}-\frac{q}{2}j_{\pi(m-2)}j_{\pi(m)}
  \\&-\frac{q}{2}j_{\pi(m)}
  -\frac{q}{2}j_{\pi(m-1)}j_{\pi(m)}
  \\
     &+\sum_{s=1}^{m}\mu_{s}
  \left(i_{s}-j_{s}\right),\\
  b_1(i)-b_2(j) = & a_1(i)-a_2(j)+\frac{q}{2}\left(i_{\pi(1)}-j_{\pi(1)}\right).
  \end{align*}
  For $u=0$, then, we have $i=j$ and
    \begin{align*}
   &a_1(i)-a_2(i)= b_1(i)-b_2(i)\\=& -\frac{q}{2}i_{\pi(m-2)}-\frac{q}{2}i_{\pi(m-2)}i_{\pi(m)}
  -\frac{q}{2}i_{\pi(m)}\\&-\frac{q}{2}i_{\pi(m-1)}i_{\pi(m)}.
  \end{align*}
  So, Eq. (\ref{eq-ACF-12}) can be reduced as
   \begin{align*}
   &2\sum_{i=0}^{2^{m}-1-u}(\xi^{a_1(\underline{i})-a_2(\underline{i})})
   \end{align*}\begin{align*}
   =&2\left(\sum_{i=0}^{2^{m-1}-1}(-1)^{i_{\pi(m-2)}}+\sum_{i=2^{m-1}}^{2^{m}-1}
   (-1)^{i_{\pi(m)}-i_{\pi(m-1)}}\right)  \\
   =&2(0+0)=0.
  \end{align*}
For $u\in \mathcal{T}_1\bigcup\mathcal{T}_2$, under the same way of the proof of Step 1, we arrive at
\begin{eqnarray}\label{eq-ACF-21}
   \sum_{i=0}^{2^{m}-1-u}\left(\xi^{a_2(i)-a_1(j)}
   +
   \xi^{b_2(i)-b_1(j)}\right)=0.
  \end{eqnarray}
  \end{description}
  Combining the cases above, we can conclude that
    \begin{eqnarray*}
   C_{\mathcal{S}_k}(u)=\rho_{\psi(\mathbf{a}_k)}(u)+\rho_{\psi(\mathbf{b}_k)}(u)=0,\\
   \hbox{for}~|u|\in \mathcal{T}_1\bigcup\mathcal{T}_2,~1\leq k\leq2,
  \end{eqnarray*}
  and    \begin{eqnarray*}
    C_{\mathcal{S}_1,\mathcal{S}_2}(u)=\rho_{\psi(\mathbf{a}_1),\psi(\mathbf{a}_2)}(u)
    +\rho_{\psi(\mathbf{b}_1),\psi(\mathbf{b}_2)}(u)=0,
   \\\hbox{for}~|u|\in \{0\}\bigcup\mathcal{T}_1\bigcup\mathcal{T}_2.
  \end{eqnarray*}
 The proof is completed.
\end{proof}

\bibliographystyle{IEEEtran}
\bibliography{IEEEfull}

\end{document}